\documentclass[onecolumn,draftclsnofoot,12pt]{IEEEtran}

\IEEEoverridecommandlockouts

\makeatletter
\newcommand*\titleheader[1]{\gdef\@titleheader{#1}}
\AtBeginDocument{%
  \let\st@red@title\@title
  \def\@title{%
    \bgroup\normalfont\large\centering\@titleheader\par\egroup
    \vskip1.5em\st@red@title}
}
\makeatother

\usepackage{epsfig}
\usepackage{url}
\usepackage{upgreek}
\usepackage{bm}
\usepackage{cite}
\usepackage{subcaption}
\usepackage{amsmath}
\usepackage{amsfonts}
\usepackage{amssymb}
\usepackage{graphicx}
\usepackage{array}
\usepackage{nicefrac}
\usepackage{textcomp}
\usepackage{xcolor}
\usepackage{tikz}
\usepackage{bm}
\usepackage{algorithmicx}
\usepackage{algorithm}
\usepackage[noend]{algpseudocode}
\usepackage{varwidth}
\usepackage{multicol}
\usepackage{mathtools}
\usepackage{cuted}
\usepackage{stfloats}
\usepackage{xr}
\usepackage{draftwatermark}
\SetWatermarkText{Draft}
\SetWatermarkScale{1}

\usepackage{scalerel}
\usetikzlibrary{svg.path}

\definecolor{orcidlogocol}{HTML}{A6CE39}
\tikzset{
  orcidlogo/.pic={
    \fill[orcidlogocol] svg{M256,128c0,70.7-57.3,128-128,128C57.3,256,0,198.7,0,128C0,57.3,57.3,0,128,0C198.7,0,256,57.3,256,128z};
    \fill[white] svg{M86.3,186.2H70.9V79.1h15.4v48.4V186.2z}
                 svg{M108.9,79.1h41.6c39.6,0,57,28.3,57,53.6c0,27.5-21.5,53.6-56.8,53.6h-41.8V79.1z M124.3,172.4h24.5c34.9,0,42.9-26.5,42.9-39.7c0-21.5-13.7-39.7-43.7-39.7h-23.7V172.4z}
                 svg{M88.7,56.8c0,5.5-4.5,10.1-10.1,10.1c-5.6,0-10.1-4.6-10.1-10.1c0-5.6,4.5-10.1,10.1-10.1C84.2,46.7,88.7,51.3,88.7,56.8z};
  }
}

\newcommand\orcidicon[1]{\href{https://orcid.org/#1}{\mbox{\scalerel*{
\begin{tikzpicture}[yscale=-1,transform shape]
\pic{orcidlogo};
\end{tikzpicture}
}{|}}}}
\usepackage{hyperref} %<--- Load after everything else
\newcounter{storeeqcounter}
\newcounter{tempeqcounter}

\newcommand{\ji}{\text{j}}
\newcommand{\R}{\text{Re}}
\newcommand{\I}{\text{Im}}
\newcommand{\NBS}{N_\text{BS}}
\newcommand{\Q}{\text{sgn}}
\newcommand{\T}{\text{T}}
\newcommand{\maxx}{\text{max}}

\newcommand{\Her}{\text{H}}
\newcommand{\diag}{\text{diag}}
\newcommand{\tr}{\text{tr}}
\newcommand{\thh}{\text{th}}
\newcommand{\BS}{\text{BS}}
\newcommand{\ext}{\text{ext}}

\DeclareMathOperator*{\argmax}{arg\,max}
\DeclareMathOperator*{\argmin}{arg\,min}

\newtheorem{theorem}{Theorem}[section]
\newtheorem{corollary}{Corollary}[theorem]
\newtheorem{lemma}[theorem]{Lemma}
\newenvironment{proof}[1]{\par\noindent\underline{Proof:}\space#1}{\hfill $\blacksquare$}

\newtheorem{remark}{Remark}

\def\BibTeX{{\rm B\kern-.05em{\sc i\kern-.025em b}\kern-.08em
    T\kern-.1667em\lower.7ex\hbox{E}\kern-.125emX}}
\begin{document}

\title{Multi-user Downlink Beamforming using Uplink Downlink Duality  with 1-bit Converters for Flat Fading Channels}

%\author{Khurram Usman~Mazher,~\IEEEmembership{Student Member,~IEEE,}
%        Amine~Mezghani,~\IEEEmembership{Member,~IEEE,}
%        and \\ ~Robert W.~Heath Jr.,~\IEEEmembership{Fellow,~IEEE}% <-this % stops a space
%\thanks{Manuscript received in February 2022.}%
%\thanks{K. U. Mazher (khurram.usman@utexas.edu) is with the Wireless Networking and Communications Group, The University of Texas at Austin, Austin, TX 78712 USA.}% <-this % stops a space
%\thanks{A. Mezghani (amine.mezghani@umanitoba.ca) is with the Department of Electrical and Computer Engineering at the University of Manitoba, Winnipeg, MB R3T 2N2, Canada.}% <-this % stops a space
%\thanks{R. W. Heath Jr. (rwheathjr@ncsu.edu) is with 6GNC, Department of Electrical and Computer Engineering at the North Carolina State University, Raleigh, NC 27695 USA.}
%\thanks{This material is based upon work supported in part by the National Science Foundation under Grant No. ECCS-1711702, in part by the Natural Sciences and Engineering Research Council of Canada (NSERC) under Grant ID RGPIN-2020-06754, and in part by gifts from Qualcomm Inc. and Samsung Research America Inc.}}% <-this % stops a space

\author{Khurram Usman~Mazher$^{\textsuperscript{\orcidicon{0000-0001-8760-391X}}}$,~\IEEEmembership{Student Member,~IEEE,}
        Amine~Mezghani$^{\textsuperscript{\orcidicon{0000-0002-7625-9436}}}$,~\IEEEmembership{Member,~IEEE,}
        and \\ ~Robert W.~Heath Jr.$^{\textsuperscript{\orcidicon{0000-0002-4666-5628}}}$,~\IEEEmembership{Fellow,~IEEE}% <-this % stops a space
\thanks{Manuscript submitted to IEEE TVT in June 2022.}%
\thanks{K. U. Mazher (khurram.usman@utexas.edu) is with the Wireless Networking and Communications Group, The University of Texas at Austin, Austin, TX 78712 USA.}% <-this % stops a space
\thanks{A. Mezghani (amine.mezghani@umanitoba.ca) is with the Department of Electrical and Computer Engineering at the University of Manitoba, Winnipeg, MB R3T 2N2, Canada.}% <-this % stops a space
\thanks{R. W. Heath Jr. (rwheathjr@ncsu.edu) is with 6GNC, Department of Electrical and Computer Engineering at the North Carolina State University, Raleigh, NC 27695 USA.}
\thanks{This work was supported in part by the National Science Foundation under Grant No. ECCS-1711702, in part by the Natural Sciences and Engineering Research Council of Canada (NSERC) under Grant ID RGPIN-2020-06754, and in part by gifts from Qualcomm Inc. and Samsung Research America~Inc.}}% <-this % stops a space

\markboth{IEEE Transactions on Vehicular Technology (Submitted paper)}{}

\maketitle

%%%%%%%%%%%%%%%%%%%%%%%%%%%%%%%%%%%%%%%%%%%%%%%%%%%%%%%%
\begin{abstract}
The increased power consumption of high-resolution data converters at higher carrier frequencies and larger bandwidths is becoming a bottleneck for communication systems. In this paper, we consider a fully digital base station equipped with 1-bit analog-to-digital (in uplink) and digital-to-analog (in downlink) converters on each radio frequency chain. The base station communicates with multiple single antenna users with individual SINR constraints. We first establish the uplink downlink duality principle under 1-bit hardware constraints under an uncorrelated quantization noise assumption. We then present a linear solution to the multi-user downlink beamforming problem based on the uplink downlink duality principle. The proposed solution takes into account the hardware constraints and jointly optimizes the downlink beamformers and the power allocated to each user. Optimized dithering obtained by adding dummy users to the true system users ensures that the uncorrelated quantization noise assumption is true under realistic settings. Detailed simulations carried out using 3GPP channel models generated from \emph{Quadriga} show that our proposed solution outperforms state of the art solutions in terms of the ergodic sum and minimum rate especially when the number of users is large. We also demonstrate that the proposed solution significantly reduces the performance gap from non-linear solutions in terms of the uncoded bit error rate at a fraction of the computational complexity.
\end{abstract}


%%%%%%%%%%%%%%%%%%%%%%%%%%%%%%%%%%%%%%%%%%%%%%%%%%%%%%%%
\section{Introduction}

Massive multiple-input-multiple-output (MIMO) promises to enhance spectral efficiency, energy efficiency, scalability, reliability, and coverage of next generation wireless cellular systems \cite{massiveMIMO,massiveMIMO2}. With massive MIMO, low complexity linear precoding techniques such as zero forcing (ZF) and maximum ratio transmission (MRT) are near optimal for a sufficiently large ratio of the number of antennas at the base station (BS) over the number of users \cite{massiveMIMO,massiveMIMO2}. The increased number of antennas and radio frequency (RF) chains, however, lead to significant increase in the power consumption, front-haul data rate requirements and hardware complexity \cite{SP,NYU,powerStuder}.

The use of low-resolution analog to digital converters (ADCs) and digital to analog converters (DACs), particularly 1-bit ADCs/DACs, is one possible solution to the excessive power consumption, higher costs and limited physical area available at the RF front-end for high carrier frequency and large bandwidth signals \cite{SP,NYU,powerStuder}. 1-bit quantization leads to a power penalty of about 2 dB for channels of practical interest with perfect or statistical channel state information (CSI) for the single user setting \cite{amine2020}. Achieving this performance in a multi-user (MU) setting can be challenging: 1) Small numbers of users lead to correlated quantization noise 2) Large numbers of active users lead to significant MU interference (MUI). In this paper, we focus on the MU downlink (DL) beamforming (BF) problem where a BS equipped with 1-bit DACs communicates with multiple single antenna users with individual signal-to-quantization-plus-interference-plus-noise ratio (SQINR) constraints. We introduce a new criterion for MU-DL precoder design under 1-bit DAC constraints based on maximizing the minimum SQINR. To the best of authors' knowledge, this optimization criterion has not been considered before for 1-bit DL transmissions.

\vspace{-0.0cm}
\subsection{Prior work}
The use of 1-bit ADCs in the MU-MIMO uplink (UL) has been studied extensively \cite {massiveMIMO1-bit,channelEstimationStuder,ULStuder,ULRobert,ULOFDMStuder}. The performance of low complexity least squares based channel estimation and linear detectors in the UL was investigated in \cite{massiveMIMO1-bit,channelEstimationStuder,ULStuder} for flat fading channels. It was shown that large sum rates can be supported despite the severe distortion caused by the 1-bit ADCs. Similar conclusions were presented for wideband channels for both single carrier and Orthogonal Frequency Division Multiplexing (OFDM) based transmissions \cite{ULRobert}. The performance of non-linear methods for frequency selective channels in the OFDM context with low resolution ADCs was studied in \cite{ULOFDMStuder}. It was concluded that the ADC resolution can be significantly reduced while achieving \emph{almost} the same performance as $\infty$-resolution ADCs when the ratio of number of BS antennas to the number of users is large. Furthermore, at large number of BS antennas to the number of users ratios, low complexity linear methods achieve virtually the same performance as non-linear methods \cite{ULOFDMStuder}.

%Prior work on MU-DL precoding under 1-bit DAC constraints can be grouped into linear and non-linear methods for various design criterion 

Linear precoding methods optimized based on various design criterion for MU-DL precoding under 1-bit DAC constraints have been proposed \cite{amine2009,amine2016,Candido,SE,amodh}. Linear precoding based on minimizing the mean square error (MMSE) of the transmitted symbols while taking into account the quantization effects was shown to outperform quantized linear Weiner filtering \cite{amine2009}. A more general MMSE framework \cite{amine2016} that also optimized the per-antenna power allocation after the 1-bit quantization operation was shown to outperform the equal per-antenna power allocation MMSE precoder \cite{amine2009}. Precoder design by minimizing the mean square error (MSE) of superposed Quadrature Phase Shift Keying (QPSK) symbols based on an iterative gradient projection algorithm has been proposed \cite{Candido}. The tradeoff between spectral efficiency and energy efficiency under ZF and MRT based precoding was studied in \cite{SE}. It was concluded that 1-bit MIMO systems need approximately $2.5$ times the number antennas to achieve the same performance as unquantized systems. Recent work using ZF precoders in settings where the quantization noise was correlated showed that significant performance improvements can be achieved with optimized dithering \cite{amodh}. The prior work on linear precoder design for 1-bit DL transmissions \cite{amine2009,amine2016,Candido,SE,amodh} is limited to minimizing the SER, the MUI (in case of ZF) and maximizing the signal power (in case of MRT). Furthermore, there is a significant performance gap compared to non-linear methods as discussed next.

Non-linear precoding methods, which map the transmit symbols directly to the quantized transmit vector, generally outperform linear precoding based techniques \cite{hela2016,hela,studer2016,Studer,studer2017}. A lookup table that mapped the transmit symbols  to the transmit signal based on minimizing the  bit error rate (BER) for each channel realization \cite{hela2016} was shown to outperform the MMSE precoders \cite{amine2009,amine2016} in terms of uncoded BER and mutual information. Similarly, a non-linear method \cite{hela} based on maximizing the safety margin (MSM) of the received symbols from the constellation decision thresholds for quantized constant envelope signals outperformed the linear precoding strategies  \cite{amine2009,amine2016,Candido}  in terms of uncoded BER. Non-linear methods based on semi-definite relaxation and squared $\ell_\infty$-norm relaxation of the SER have also been proposed \cite{studer2016,Studer}. Another non-linear method based on solving the biconvex relaxation of the SER for 1-bit systems using alternating minimization was proposed in \cite{studer2017}. The performance of all these \cite{hela2016,hela,studer2016,Studer,studer2017} non-linear methods is comparable. These methods, however, result in a significant computational cost for systems with larger dimensions due to exponential increase in their complexity with the increase in number of users, the constellation size and the number of antennas \cite{Studer}. Additionally, an optimization problem has to be solved at every time step during the coherence time for each transmit symbol vector. Lastly, most of the non-linear methods \cite{hela2016,hela,studer2016,Studer,studer2017} have hyperparameters that need to be appropriately chosen according to the operating conditions. There has also been work in 1-bit DL precoding for frequency selective channels using OFDM based formulations \cite{StuderOFDMconf,StuderOFDM}. In this paper, we restrict ourself to the flat fading channel setting leaving the generalization to wideband channels as future work.

The design criterion in the prior work are focused on minimizing the BER, the SER, the MUI (in case of ZF) or maximizing the signal power (in case of MRT) and symbol safety margins. Additionally, most of the numerical/analytical results were obtained on Rayleigh fading channels with independent and identically distributed (IID) entries. In this paper, we introduce a new design criterion for MU-DL precoding under 1-bit DAC constraints that maximizes the minimum SQINR and compare the proposed solution to the existing work using realistic channel models.

%The prior is also limited to settings with 1-bit DACs and assumes $\infty-$resolution ADCs which can have an effect on the quality of channel estimates available and UL transmissions.

\vspace{-0.3cm}
\subsection{Contributions}
In this work, we provide a linear precoding based solution to the MU-DL-BF and power allocation problem under 1-bit hardware constraints at the BS. The BS communicates with multiple single antennas users with individual SQINR constraints over flat fading channels. The classical solution to the MU-DL-BF problem with $\infty$-resolution ADCs and DACs makes use of the UL-DL duality to solve the problem by an iterative alternating minimization procedure \cite{MUDL}. The same procedure can not be applied to the \emph{Bussgang} decomposition \cite{StuderOFDM,SE} based linearized version of the 1-bit problem because of the different quantization noise in the UL and DL SQINRs.
%
% We first prove that the UL-DL duality principle holds for the case with 1-bit hardware constraints provided that the quantization noise is uncorrelated. 
% 
% We then generalize the alternating minimization solution to the MU-DL-BF problem to the~setting where the BS is constrained by 1-bit ADCs/DACs. 
% 
% Next, we generalize the UL-DL duality principle and the alternating minimization based solution to the setting where the quantization noise is correlated by adding controlled noise to the system. Optimized dithering is added through dummy users which lie in the null space of the true users of the system. 
% 
% Our simulation results obtained on 3GPP channel models generated by Quadriga \cite{Quadriga} demonstrate that the proposed solution outperforms state of the art linear solution \cite{amodh} by more than 6 b/s/Hz in terms of the ergodic sum and minimum rate. We also demonstrate that the proposed solution achieves performance comparable with non-linear precoding methods \cite{Studer} in terms of uncoded BER. We conclude the paper by commenting on how the proposed analysis and framework can be extended to other type of hardware constraints such as constant envelope \cite{amodh}. 
% 
 The main contributions of this paper are:

\begin{itemize}
\item We establish the UL-DL duality principle under 1-bit ADC/DAC constraints. This is different from its ideal $\infty$-resolution counterpart because the quantization noise in the SQINR expressions is a function of the channel realizations in the UL and the precoders in the DL. We show that the well established UL-DL duality result \cite{MUDL} can be generalized to the 1-bit setting under the assumption that the quantization noise is uncorrelated.

\item We propose an alternating minimization based solution to the MU-DL-BF problem that incorporates the 1-bit constraints and jointly optimizes the DL power allocation vector and the DL BF matrix based on maximizing the minimum SQINR across all users. The proposed solution makes use of the UL-DL duality result to break the larger MU-DL-BF problem into smaller sub-problems.

\item We generalize the UL-DL duality principle and the proposed algorithm to settings where the quantization noise is correlated by introducing optimized dithering. Optimized dithering is added under the per-user SQINR constraint framework by adding dummy users to the system that lie in the null space of the true user channels. 

%\item We present our results based on realistic channel realizations obtained from Quadriga \cite{Quadriga}. This is in contrast to most if not all of the prior work based on Rayleigh fading channels with IID entries.

\item We present results based on channel realizations drawn from 3GPP channel models using Quadriga \cite{Quadriga}. Our numerical results show significant performance improvement over state of the art linear methods both in terms of ergodic sum rate and ergodic minimum rate. We also demonstrate that the proposed solution achieves performance comparable with non-linear precoding methods \cite{Studer} in terms of uncoded BER.
%\item Our results also show significant gains for the setting when the BS only has knowledge of statistical CSI obtained from a co-located sub-6 GHz system.
\end{itemize}

In our prior work \cite{pw},  we published initial results from the detailed study carried out in this paper. After stating the UL-DL duality principle under the uncorrelated quantization noise assumption, we compared the performance of the proposed MU-DL-BF algorithm to ZF based linear precoding. We concluded that the proposed solution outperforms ZF based precoding in terms of the ergodic sum rate, especially when the number of users is large. The present paper, however, is a significant extension of \cite{pw}. In addition to the formal proof of the UL-DL duality principle under 1-bit constraints, we also provide the crucial details of the joint power allocation and beamforming solution (such as the feasibility of the MU-DL-BF problem and convergence of the proposed alternating minimization solution) that were missing in our prior work. The uncorrelated quantization noise assumption was not justified in our prior work \cite{pw}. In this paper, we introduce optimized dithering (through dummy users) in the signal before quantization and generalize the presented ideas to the setting where the uncorrelated quantization noise assumption does not hold. As a result, the solution presented in this paper outperforms our initial results published in \cite{pw} for a smaller number of active users. Our prior work \cite{pw} also lacked any comparison with a state-of-the-art non-linear precoding method. Herein, we conduct detailed numerical experiments that include comparison with SQUID \cite{Studer} in terms of uncoded BER and further comment on the robustness of the proposed solution to channel estimation errors.

The rest of this paper is organized as follows. In Section \ref{sec:systemModel}, we describe the system model and the small angle approximation based on which the quantization noise becomes uncorrelated. In Section \ref{sec:duality}, we establish the UL-DL duality principle under the uncorrelated quantization noise approximation. Next in Section \ref{sec:solution}, we provide the details of the joint power allocation and beamforming optimization algorithm for the MU-DL-BF problem. In Section \ref{sec:dummy}, we extend the proposed solution to settings where the uncorrelated quantization noise assumption does not hold. We present numerical results in Section \ref{sec:results} before concluding the paper in Section \ref{sec:conc}.

\emph{Notation:} $\mathbf{B}$ is a matrix, $\mathbf{b}$ is a vector and $b$ is a scalar. The operator $(\cdot)^\T$, $(\cdot)^\Her$, and $(\cdot)^\ast$ denote the transpose, conjugate transpose and conjugate of a matrix/vector. $\mathrm{diag} (\mathbf{B})$ denotes a diagonal matrix containing only the diagonal elements of $\mathbf{B}$. $\mathrm{tr}(\mathbf{B})$ denotes the trace of matrix $\mathbf{B}$. $\| \mathbf{B} \|_F$ denotes the Frobenius norm of $\mathbf{B}$. $\mathbf{I}_N$ represents the identity matrix of size $N \times N$. The vector $\mathbf{1}_N$ ($\mathbf{0}_N$) denotes a vector of all ones (zeros) of length $N$. The matrix $\mathbf{R}_{\mathbf{b}}$ denotes the covariance matrix of the signal $\text{b}$. $\lambda_\maxx (\mathbf{B})$ denotes the dominant eigenvalue of $\mathbf{B}$. $\mathbb{N}(\mathbf{B})$ denotes the nullspace of the matrix $\mathbf{B}$. $\|\mathbf{b}\|_p$ is the $p$-norm of $\mathbf{b}$. $\mathbf{b} \geq b$ denotes entry-wise comparison between $\mathbf{b}$ and $b$. $\mathbf{e}_k$ denotes the canonical basis vector with a 1 at the $k^\thh$ index and zeros elsewhere. The function $\mathrm{sgn}(a)$ denotes the signum function applied component-wise to the $\R(a)$ (real) and $\I(a)$ (imaginary) parts of $a$. The notations $|\cdot| , {(\cdot)}^k$ and $ \angle(\cdot)$ denote the absolute value, $k^{\thh}$ power and phase operation applied to a scalar or element-wise to a vector/matrix. $\mathcal{CN}( \bm{\mu} , \bm{\Sigma} )$ denotes a complex Gaussian multi-variate distribution with mean  $\bm{\mu}$ and covariance $\bm{\Sigma}$. 

%$\mathbf{B}_i$ and $\mathbf{B}_{ij}$ denotes the $i^{\thh}$ row and $i^{\thh} , j^{\thh}$ entry of the matrix $\mathbf{B}$. $\mathbf{b}_i$ denote the $i^{\thh}$ entry of $\mathbf{b}$.

%%%%%%%%%%%%%%%%%%%%%%%%%%%%%%%%%%%%%%%%%%%%%%%%%%%%%%%%
\vspace{-0.2cm}
\section{System model} \label{sec:systemModel}
We consider a DL scenario where a single BS with $\NBS$ antennas and RF chains, each equipped with a 1-bit DAC, communicates with $K$ single antenna users, as illustrated in Fig. \ref{fig:system}. The BS sends IID  $\mathcal{N}( 0, 1 )$ signals $s_k$ for $1 \leq k \leq K$ to the users. The channel from the BS to the $k^{\thh}$ user at time $t$ is assumed to be frequency flat and is denoted by $\mathbf{h}_k(t) \in \mathbb{C}^{N_\text{BS} \times 1}$. ${\mathbf{H}} \in \mathbb{C}^{K \times N_\text{BS}}$ denotes the collective channel matrix for the $K$ users. During the DL stage, the symbol $s_k$ is mapped to the antenna array using the unit-norm beamformer $\mathbf{t}_k  \in \mathbb{C}^{N_\text{BS} \times 1}$. The beamformers corresponding to all $K$ users can be collected in a matrix $\mathbf{T} = [\mathbf{t}_1 , \dots , \mathbf{t}_K]$. During the DL stage, the BS has a total transmit power constraint of $P_\BS$ Watts. We assume that the BS and all users operate in the same time-frequency resource and are synchronized. 
%We assume that an estimate of the channel, $\widehat{\mathbf{h}}_k(t)$, corrupted by additive IID Gaussian noise with covariance matrix $\sigma^2 \mathbf{I}$ is available at the BS. and $\widehat{\mathbf{H}} \in \mathbb{C}^{K \times N_\text{BS}}$ and estimated channel

We also consider the corresponding UL scenario where the $K$ users send information symbols to the BS. Chances are that the BS would have a 1-bit ADC on each RF chain in the UL for reasons similar to why it had a 1-bit DAC on each RF chain in the DL. This, however, is not a requirement and can be thought of as a mathematical construct for the purpose of this paper. During the UL stage, the symbol $s_k$ is received at the BS using the unit-norm beamforming vector $\mathbf{u}_k$. The beamformers corresponding to all $K$ users are collected in a matrix $\mathbf{U} = [\mathbf{u}_1 , \dots , \mathbf{u}_K]$. The $K$ users transmit under a \emph{sum} power constraint of $P_\BS$ Watts equal to the total power constraint of the BS during the DL stage. We emphasize here that the users do not actually transmit under the sum power constraint in the UL. This is only a conceptual construction and will be used in Section \ref{sec:joint} to break down the MU-DL-BF problem into smaller sub-problems.

We conclude this section by describing the small angle approximation which essentially says that the off-diagonal entries of the covariance matrix of the signal before quantization are small compared to the diagonal entries. This makes the quantization noise uncorrelated which, as will be shown in Section \ref{sec:duality}, is crucial for proving UL-DL duality in the presence of hardware constraints.

%We evaluate the performance of the proposed algorithm in settings where the BS only has knowledge of channel statistics in Section \ref{sec:results}. 

\begin{figure}[t]
    	\begin{center}
    		\includegraphics[width=.48\textwidth,clip,keepaspectratio]{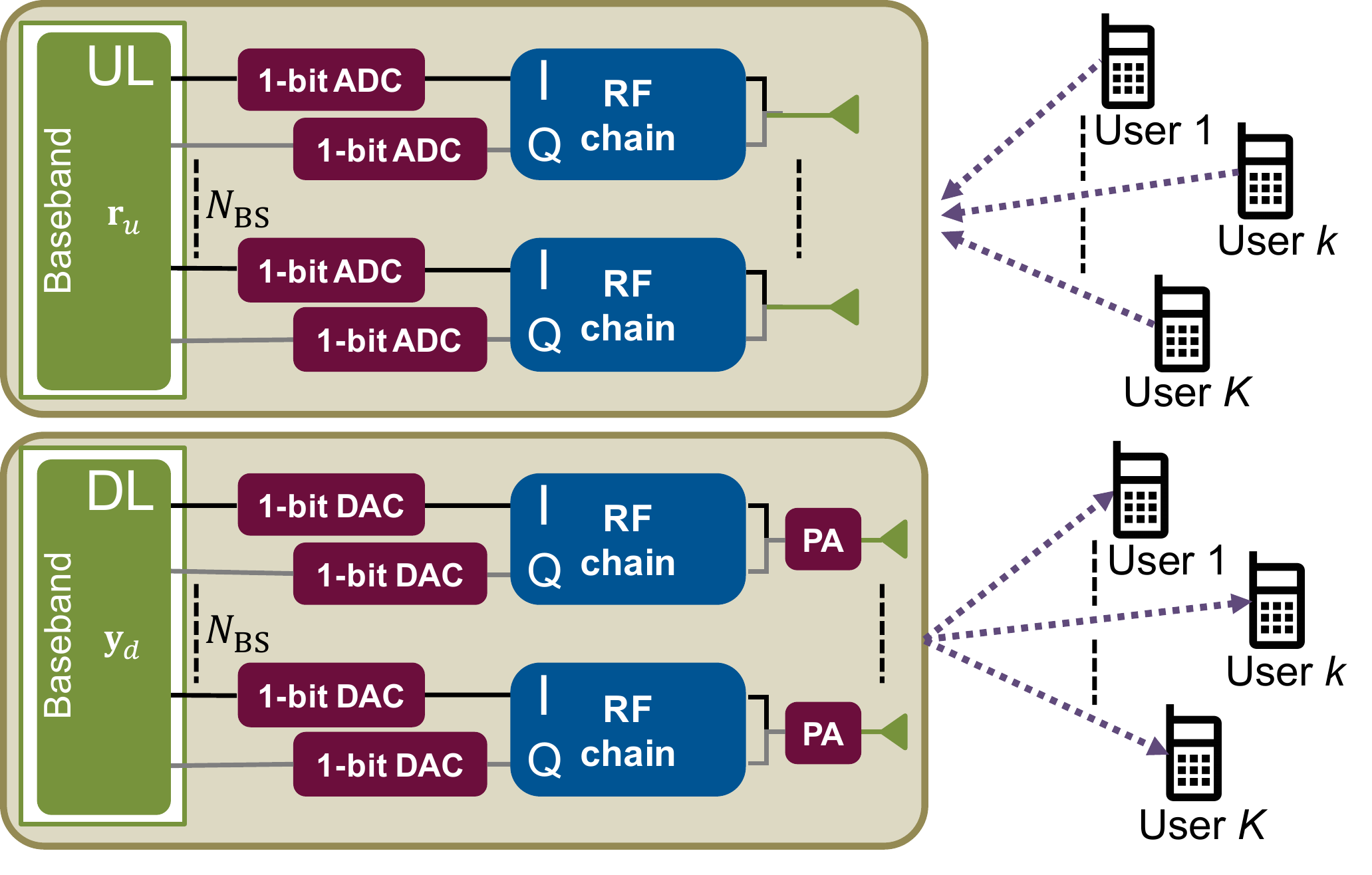}
    	\end{center}
\vspace{-0.2cm}
    	\caption{Functional block diagram of the UL and DL system model where a fully digital BS with $N_\BS$ antennas and 1-bit ADC/DAC on each RF chain communicates with $K$ single antenna users.}
    	 \label{fig:system}
\vspace{-0.5cm}
\end{figure}

\vspace{-0.2cm}
\subsection{Downlink SQINR} \label{sec:DL}
With $\mathbf{q} = [q_1 , \dots , q_K]^\T$ denoting the DL power allocation vector, the signal transmitted by the BS during the DL stage can be written as $\mathbf{y_\text{d}} = \sum_{k = 1}^K \sqrt{q_k} \mathbf{t}_k^\ast s_k$. The signal after the 1-bit DAC is given by $\mathbf{r_\text{d}} = \frac{ \Q(\mathbf{y_\text{d}})}{\sqrt{2}}$. Let $\mathbf{R}_{\mathbf{y_\text{d}}} = \sum_{k = 1}^K {q_k} \mathbf{t}_k^\ast \mathbf{t}_k^\T$ denote the covariance matrix of the DL signal $\mathbf{y_\text{d}}$ before the 1-bit quantization operation. The \emph{Bussgang} theorem \cite{SE} can be used to decompose the signal into a useful linear part and an uncorrelated distortion $\pmb{\eta}_\text{d}$ with the covariance  $\mathbf{R}_{\pmb{\eta}_\text{d}}$. With $\mathbf{A}_\text{d} = \sqrt{\frac{2}{\pi}} \diag \left( \mathbf{R}_{\mathbf{y_\text{d}}} \right)^{-\frac{1}{2}}$ denoting the Bussgang gain, the signal after the 1-bit DAC can be rewritten as
\begin{equation} \label{eq:downlinkSig_r1}
 \mathbf{r_\text{d}} = \mathbf{A}_\text{d} \sum_{k = 1}^K \sqrt{q_k} \mathbf{t}_k^\ast s_k + \pmb{\eta}_\text{d}.
 \vspace{-0.1cm}
\end{equation}
Any \emph{per-user} power allocation done before the 1-bit DAC will be completely wiped out due to the quantization operation. Power allocation is done again on a \emph{per-antenna} basis after the 1-bit DAC operation. This is achieved by multiplication with the non-negative diagonal matrix $\mathbf{Q}$. The total per-antenna power allocation is constrained to be equal to the DL transmit power given by forcing $P_\BS = \text{tr}(\mathbf{Q}\mathbf{Q}^\Her)$. The linearized signal received at the $k^\thh$ user corrupted by IID $\mathcal{N}( 0 , \sigma^2 )$ noise $n$ is
\begin{equation} \label{eq:downlinkSig}
 \vspace{-0.1cm}
{y_{\text{d},k}} = \mathbf{h}_k^\T \mathbf{Q} \mathbf{A}_\text{d} \sum_{k = 1}^K \sqrt{q_k} \mathbf{t}_k^\ast s_k + {n} + \mathbf{h}_k^\T \mathbf{Q} \pmb{\eta}_\text{d}.
\end{equation}
Using (\ref{eq:downlinkSig}) and defining $\mathbf{R}_k = \mathbf{h}_k  \mathbf{h}_k^\Her$, the DL SQINR for the $k^\thh$ user, $\gamma_k^{\text{DL}}(\mathbf{T} ,  \mathbf{Q} , \mathbf{q})$, is given by (\ref{eq:downlinkSINR_covar}) at the bottom of the page.
%\begin{equation}
%\gamma_k^{\text{DL}}(\mathbf{T} ,  \mathbf{Q} , \mathbf{q}) = \frac{ q_k \mathbf{t}_k^\Her \mathbf{A}_\text{d} \mathbf{Q} \mathbf{R}_k \mathbf{Q}^\Her \mathbf{A}_\text{d}^\Her  \mathbf{t}_k}{  \sum_{ \substack{i = 1\\ i \neq k}}^K \underbrace{ q_i \mathbf{t}_i^\Her \mathbf{A}_\text{d} \mathbf{Q} \mathbf{R}_k \mathbf{Q}^\Her \mathbf{A}_\text{d}^\Her  \mathbf{t}_i }_{\text{MUI}} + \underbrace{ \sigma^2 }_{\text{IID}} + \underbrace{ \tr \left( \mathbf{Q} \mathbf{R}_{\pmb{\eta}_\text{d}} \mathbf{Q}^\Her \mathbf{R}_k^\ast \right)}_{\text{QN}} } .
%\label{eq:downlinkSINR_covar}
%\end{equation}
The DL SQINR for the $k^\thh$ user in (\ref{eq:downlinkSINR_covar}) is a function of the beamformer matrix $\mathbf{T}$, the power allocation matrix $\mathbf{Q}$, and the power allocation vector $\mathbf{q}$. %The first, second, and the third term in the denominator correspond to the MUI, IID noise, and quantization noise (QN). 

\addtocounter{equation}{1}
\setcounter{storeeqcounter}{\value{equation}}

\begin{figure*}[b]
\hrulefill
\normalsize
\setcounter{tempeqcounter}{\value{equation}} % temp store of current value
\begin{IEEEeqnarray}{rCl}
\setcounter{equation}{\value{storeeqcounter}} % number of this equation
\gamma_k^{\text{DL}}(\mathbf{T} ,  \mathbf{Q} , \mathbf{q}) = \frac{ q_k \mathbf{t}_k^\Her \mathbf{A}_\text{d} \mathbf{Q} \mathbf{R}_k \mathbf{Q}^\Her \mathbf{A}_\text{d}^\Her  \mathbf{t}_k}{  \sum_{ \substack{i = 1\\ i \neq k}}^K \underbrace{ q_i \mathbf{t}_i^\Her \mathbf{A}_\text{d} \mathbf{Q} \mathbf{R}_k \mathbf{Q}^\Her \mathbf{A}_\text{d}^\Her  \mathbf{t}_i }_{\text{MUI}} + \underbrace{ \sigma^2 }_{\text{IID}} + \underbrace{ \tr \left( \mathbf{Q} \mathbf{R}_{\pmb{\eta}_\text{d}} \mathbf{Q}^\Her \mathbf{R}_k^\ast \right)}_{\text{QN}} }.
\label{eq:downlinkSINR_covar}
\end{IEEEeqnarray}
\setcounter{equation}{\value{tempeqcounter}} % restore correct value
\end{figure*}

%%%%%%%%%%%%%%%%%%
\vspace{-0.2cm}
\subsection{Uplink SQINR} \label{sec:UL}
With $\mathbf{n}$ denoting the IID Gaussian noise with covariance $\sigma^2 \mathbf{I}$ and $p_k$ the TX power of the $k^\thh$ user, the signal received at the BS during the UL stage can be written as $\mathbf{y_\text{u}} = \sum_{k = 1}^K \sqrt{p_k}  \mathbf{h}_k s_k + \mathbf{n}$. The signal after 1-bit sampling is given by $\mathbf{r_\text{u}} = \frac{ \Q(\mathbf{y_\text{u}})}{\sqrt{2}}$. For $\mathbf{R}_{\mathbf{y_\text{u}}} = \left( \sum_{k = 1}^K p_k  \mathbf{R}_k + \sigma^2 \mathbf{I} \right)$, define $\mathbf{A}_\text{u} = \sqrt{\frac{2}{\pi}} \diag \left( \mathbf{R}_{\mathbf{y_\text{u}}} \right)^{-\frac{1}{2}}$ to be the Bussgang gain. Like its DL counterpart, the UL signal can be decomposed into a linear signal part and an uncorrelated distortion $\pmb{\eta}_\text{u}$ with covariance $ \mathbf{R}_{\pmb{\eta}_\text{u}}$ using the Bussgang decomposition as
\begin{equation} \label{eq:uplinkSig_r2}
 \vspace{-0.0cm}
 \mathbf{r_\text{u}} = \mathbf{A}_\text{u} \sum_{k = 1}^K \sqrt{p_k}  \mathbf{h}_k s_k + \mathbf{A}_\text{u} \mathbf{n} + \pmb{\eta}_\text{u}.
\end{equation}
The signal from the $k^\thh$ user after combining using the beamformer $\mathbf{u}_k$ is given by
\begin{equation} \label{eq:uplinkSig_r22}
{y_{\text{u},k}} = \mathbf{u}_k^\Her \mathbf{A}_\text{u} \sum_{k = 1}^K \sqrt{p_k}  \mathbf{h}_k s_k + \mathbf{u}_k^\Her \mathbf{A}_\text{u} \mathbf{n} + \mathbf{u}_k^\Her \pmb{\eta}_\text{u}.
 \vspace{-0.0cm}
\end{equation}
Using the linearized model in (\ref{eq:uplinkSig_r22}) and the UL power allocation vector $\mathbf{p} = [p_1 , \dots , p_K]^\T$,  the UL SQINR obtained by using the linear combiner $\mathbf{u}_k$ for the $k^\thh$ user, $\gamma_k^{\text{UL}}(\mathbf{u}_k ,  \mathbf{p})$, is given~by 
\begin{equation}
\gamma_k^{\text{UL}}(\mathbf{u}_k ,  \mathbf{p}) = \frac{p_k \mathbf{u}_k^\Her \mathbf{A}_\text{u}^\Her \mathbf{R}_k \mathbf{A}_\text{u} \mathbf{u}_k}{  \mathbf{u}_k^\Her \bigg(  \sum_{ \substack{i = 1\\ i \neq k }}^K \underbrace{ p_i \mathbf{A}_\text{u}^\Her \mathbf{R}_i \mathbf{A}_\text{u} }_{\text{MUI}} + \underbrace{ \sigma^2 \mathbf{A}_\text{u}^\Her \mathbf{A}_\text{u} }_{\text{IID}} + \underbrace{ \mathbf{R}_{ \pmb{\eta}_\text{u}} }_{\text{QN}} \bigg) \mathbf{u}_k }.
\label{eq:uplinkSINR_covar}
\end{equation}
%The first, second, and the third term in the denominator correspond to the MUI, IID noise, and quantization noise. 
Unlike the DL SQINR in (\ref{eq:downlinkSINR_covar}) which depends on the power allocation vector $\mathbf{q}$ and the beamformer matrix $\mathbf{T}$ of all $K$ users, the UL SQINR in (\ref{eq:uplinkSINR_covar}) for the $k^\thh$-user depends only on the power allocation vector $\mathbf{p}$ and the combiner for the $k^\thh$-user $\mathbf{u}_k$. This observation will become important in Section \ref{sec:joint} where we recast the original MU-DL-BF problem in terms of the `easier' MU-UL-BF problem by making use of the UL-DL duality proved in Section \ref{sec:duality}. This allows us to decouple the bigger MU-DL-BF problem into $K$ smaller subproblems which can each be solved separately.

The linearized DL and UL SQINRs in (\ref{eq:downlinkSINR_covar}) and (\ref{eq:uplinkSINR_covar}) differ from their $\infty$-resolution counterparts in the quantization noise covariance matrices, $\mathbf{R}_{ \pmb{\eta}_\text{d}}$ and $\mathbf{R}_{ \pmb{\eta}_\text{u}}$, introduced in the denominators. $\mathbf{R}_{ \pmb{\eta}_\text{u}}$ depends on the channel realization whereas  $\mathbf{R}_{ \pmb{\eta}_\text{d}}$ is a function of the BF matrix. The solution to the MU-DL-BF problem with per-user SINR constraints for $\infty$-resolution ADCs/DACs makes use of the UL-DL duality principle \cite{MUDL}. The general UL-DL duality does not hold in the presence of 1-bit hardware constraints due to the quantization noise matrices in the UL and DL SQINRs. In Section \ref{sec:duality}, we show that the UL-DL duality principle can in fact be generalized to incorporate 1-bit hardware constraints under certain conditions and then use it to solve the 1-bit MU-DL-BF problem.

%where
%\begin{equation} \label{eq:downlinkSig_r2}
%\mathbf{A}_\text{d} = \sqrt{\frac{2}{\pi}} \diag \left( \mathbf{R}_{\mathbf{y_\text{d}}} \right)^{-\frac{1}{2}} = \sqrt{\frac{2}{\pi}} {\diag} \left( \sum_{k = 1}^K {q_k} \mathbf{t}_k^\ast \mathbf{t}_k^\T  \right)^{-\frac{1}{2}}.
%\end{equation}

%Having obtained the linearized DL/UL SQINRs, we now use (\ref{eq:uplinkSINR_covar}) and (\ref{eq:downlinkSINR_covar}) to show that the same individual SQINR constraints (or a multiplicative factor) can be met in both DL and UL using the same set of optimized precoders/combiners and separately optimized DL/UL power allocations.

%\begin{equation} \label{eq:downlinkSINR_covar}
%\gamma_k^{\text{DL}}(\mathbf{T} ,  \mathbf{Q}) = \cfrac{ \mathbf{t}_k^H \mathbf{A}_\text{d} \mathbf{Q} \mathbf{R}_k \mathbf{Q}^H \mathbf{A}_\text{d}^H  \mathbf{t}_k}{  \sum_{ \substack{i = 1\\ i \neq k}}^K \mathbf{t}_i^H \mathbf{A}_\text{d} \mathbf{Q} \mathbf{R}_k \mathbf{Q}^H \mathbf{A}_\text{d}^H  \mathbf{t}_i + \sigma^2 + \text{tr}(\mathbf{Q} \mathbf{R}_{\pmb{\eta}_\text{d}} \mathbf{Q}^H \mathbf{R}_k)}.
%\end{equation}

%%%%%%%%%%%%%%%%%%
\vspace{-0.2cm}
\subsection{Small angle approximation} \label{sec:uqn}
The uncorrelated distortion as a result of applying Bussgang decomposition on the DL and UL signals in (\ref{eq:downlinkSig_r1}) and (\ref{eq:uplinkSig_r2}), $\pmb{\eta}_\text{p}$ for $\text{p} \in \{ \text{d,u} \}$, is
\begin{equation} \label{eq:upn1}
\pmb{\eta}_\text{p} =  \mathbf{r_\text{p}}  - \mathbf{A}_\text{p}  \mathbf{y_\text{p}}.
\end{equation}
The covariance matrix of the distortion is
\begin{equation} \label{eq:upn2}
\mathbf{R}_{\pmb{\eta}_\text{p}} =  \mathbf{R}_{ \mathbf{r_\text{p}} }  - \mathbf{A}_\text{p}  \mathbf{R}_{\mathbf{y_\text{p}}} \mathbf{A}_\text{p}^\Her.
\end{equation}
Define $\mathbf{X}_{\text{p}} =  \R \left( \diag \left( \mathbf{R}_{\mathbf{y_\text{p}}} \right)^{-\frac{1}{2}}  \mathbf{R}_{\mathbf{y_\text{p}}} \diag \left( \mathbf{R}_{\mathbf{y_\text{p}}} \right)^{-\frac{1}{2}} \right)$ and $\mathbf{Y}_{\text{p}} =  \I \left( \diag \left( \mathbf{R}_{\mathbf{y_\text{p}}} \right)^{-\frac{1}{2}}  \mathbf{R}_{\mathbf{y_\text{p}}} \diag \left( \mathbf{R}_{\mathbf{y_\text{p}}} \right)^{-\frac{1}{2}} \right)$ to be the real and imaginary parts of the correlation coefficients of the signal $\mathbf{y_\text{p}}$ before 1-bit quantization. It can be shown \cite{SE} that the correlation matrix of the quantized signal $\mathbf{r_\text{p}}$ is given by
\begin{equation} \label{eq:upn3}
\mathbf{R}_{ \mathbf{r_\text{p}} }  = \frac{2}{\pi}\left( \text{sin}^{-1}\left(\mathbf{X}_{\text{p}}\right) + \ji \text{ sin}^{-1}\left(\mathbf{Y}_{\text{p}}\right) \right),
\end{equation}
where the $\text{sin}^{-1}$ operation is applied element-wise on the matrices $\mathbf{X}_{\text{p}}$ and $\mathbf{Y}_{\text{p}}$. $\mathbf{R}_{\pmb{\eta}_\text{p}}$ is then given by
\begin{equation} \label{eq:upn4}
\begin{aligned}
\mathbf{R}_{\pmb{\eta}_\text{p}} & = \frac{2}{\pi}\left( \text{sin}^{-1}\left(\mathbf{X}_{\text{p}}\right) + \ji \text{ sin}^{-1}\left(\mathbf{Y}_{\text{p}}\right) \right) - \frac{2}{\pi} \left( \left(\mathbf{X}_{\text{p}}\right) + \ji \left(\mathbf{Y}_{\text{p}}\right) \right). \\
\end{aligned}
\end{equation}
We now make the approximation that $\text{sin}^{-1}(x) = x + o(x^3)$ for $|x| < 1$. This approximation is equivalent to a first-order Taylor expansion of the  $\text{sin}^{-1}(\cdot)$ function and becomes more accurate for the case of high number of users in DL or in the low SNR regime in UL \cite{amodh}. Under this approximation, the matrix $\mathbf{R}_{\pmb{\eta}_\text{p}} = \left( 1 - \frac{2}{\pi} \right)\mathbf{I}$. The $\text{sin}^{-1}(x) \approx x$ approximation thus makes the quantization noise $\pmb{\eta}_\text{p}$ uncorrelated. This approximation will be used in Section \ref{sec:duality} for proving UL-DL duality under 1-bit ADC and DAC constraints. 

%\begin{equation} \label{eq:uplinkSig_r3}
%\begin{aligned}
%\mathbf{R}_{\pmb{\eta}_\text{p}} & \approx \left( 1 - \frac{2}{\pi} \right)\mathbf{I}.  REDUCTION \\
%\end{aligned}
%\end{equation}

%The off-diagonal elements of $\mathbf{X}_{\text{p}}$ and $\mathbf{Y}_{\text{p}}$ will always be smaller than 1, thus further justifying this approximation.
% Fig. \ref{fig:approx} provides numerical proof that this approximation is not too far off from the exact $\text{sin}^{-1}$ operation.
%\begin{figure}[t]
%    	\begin{center}
%    		\includegraphics[width=.5\textwidth,clip,keepaspectratio]{figsK/aSin_approx.eps}
%    	\end{center}
%    	\caption{$\text{sin}^{-1}(x) \approx x$ approximation for $x < 1$. The blue curve illustrates the exact $\text{sin}^{-1}$ operation whereas the red curve depicts the approximation we used. It can be seen that the approximation holds true for a significant portion of the argument $x$.}
%    	 \label{fig:approx}
%\end{figure}

%%%%%%%%%%%%%%%%%%%%%%%%%%%%%%%%%%%%%%%%%%%%%%%%%%%%%%%%
%%%%%%%%%%%%%%%%%%
\vspace{-0.2cm}
\section{UL-DL duality with hardware constraints} \label{sec:duality}

In this section, we show that the UL-DL duality is preserved for a system with 1-bit DACs/ADCs under the small angle approximation introduced in Section \ref{sec:uqn}. By writing the linearized DL and UL SQINR expressions (\ref{eq:downlinkSINR_covar}) and (\ref{eq:uplinkSINR_covar}) in equivalent matrix form, we show that the same SQINR  constraints (or a multiplicative factor) can be achieved in both DL and UL by appropriately relating the linear beamforming and combining matrices and separately optimized DL/UL power allocation vectors (for the same total power constraint). This result is summarized in Theorem \ref{theorem:theorem1}.

\begin{theorem} \label{theorem:theorem1}
Consider a BS equipped with 1-bit DACs communicating with $K$ users with target SQINR values $\{\gamma_k\}$ for $k \in \{1, \dots K\}$ using the BF matrix $\mathbf{T}$, DL power allocation vector $\mathbf{q}$ and per-antenna power allocation matrix $\mathbf{Q} = \diag \left(\sum_{i=1}^K q_i {\mathbf{t}}_i {\mathbf{t}}_i^\Her \right)^{\frac{1}{2}}$. It can be shown that the same set of SQINR values can be achieved in the UL under 1-bit ADC constraints by letting ${\mathbf{t}}_k = \mathbf{A}_\text{u} \mathbf{u}_k / \| \mathbf{A}_\text{u} \mathbf{u}_k \|_2$ and $\mathbf{p} =  \frac{\pi \sigma^2}{2} \left(  \mathbf{I}_K -  \mathbf{D}({\mathbf{T}}) \mathbf{\Psi}^\T({\mathbf{T}}) \right)^{-1} \mathbf{D}({\mathbf{T}}) \mathbf{1}_K$ under the same sum power constraint in the UL stage as the BS transmit power in the DL stage.
\end{theorem}

\begin{proof}
See Sections \ref{sec:dualityDL} and \ref{sec:dualityUL}.
\end{proof}

%%%%
%\addtocounter{equation}{1}
%\setcounter{storeeqcounter}{\value{equation}}
%
%\begin{figure*}[b]
%\hrulefill
%\normalsize
%\setcounter{tempeqcounter}{\value{equation}} % temp store of current value
%\begin{IEEEeqnarray}{rCl}
%\setcounter{equation}{\value{storeeqcounter}} % number of this equation
%\gamma_k = \frac{ q_k \mathbf{t}_k^\Her \mathbf{R}_k \mathbf{t}_k}{  \sum_{ \substack{i = 1\\ i \neq k}}^K q_i \mathbf{t}_i^\Her \mathbf{R}_k \mathbf{t}_i + \frac{\pi}{2} \sigma^2 + \left( \frac{\pi}{2} - 1 \right) \tr \left( \diag \left(\sum_{i=1}^K q_i {\mathbf{t}}_i {\mathbf{t}}_i^\Her \right) \mathbf{R}_k^\ast \right)}.
%\label{eq:downlinkSINR_covar1}
%\end{IEEEeqnarray}
%\setcounter{equation}{\value{tempeqcounter}} % restore correct value
%\end{figure*}

\addtocounter{equation}{1}
\setcounter{storeeqcounter}{\value{equation}}

\begin{figure*}[b]
\vspace{-0.5cm}
\hrulefill
\normalsize
\setcounter{tempeqcounter}{\value{equation}} % temp store of current value
\begin{IEEEeqnarray}{rCl}
\setcounter{equation}{\value{storeeqcounter}} % number of this equation
\gamma_k = \frac{ q_k \mathbf{t}_k^\Her \mathbf{R}_k \mathbf{t}_k}{  \sum_{ \substack{i = 1\\ i \neq k}}^K q_i \mathbf{t}_i^\Her \mathbf{R}_k \mathbf{t}_i + \frac{\pi}{2} \sigma^2 + \left( \frac{\pi}{2} - 1 \right) \tr \left( \sum_{i=1}^K q_i \mathbf{t}_i^\Her \diag(\mathbf{R}_k^\ast) \mathbf{t}_i  \right)}.
\label{eq:downlinkSINR_covar2}
\end{IEEEeqnarray}
\setcounter{equation}{\value{tempeqcounter}} % restore correct value
\end{figure*}
%%%%

%%%%%%%%%%%%%%%%%%%%
\vspace{-0.2cm}
\subsection{Downlink SQINR} \label{sec:dualityDL}
Let us select $\mathbf{Q} = \diag \left(\sum_{k=1}^K q_k {\mathbf{t}}_k {\mathbf{t}}_k^\Her \right)^{\frac{1}{2}}$ as the DL per-antenna power allocation matrix. With this choice of $\mathbf{Q}$, the per-antenna power allocation after the quantization operation results in the same power as for a $\infty-$resolution DAC~using the same beamformers. Let $\gamma_k$ denote the target DL SQINR~for user $k$ that needs to be achieved by appropriately choosing the DL power allocation vector $\mathbf{q}$ . Equating the target DL SQINR $\gamma_k$ to the achieved DL SQINR $\gamma_k^{\text{DL}}(\mathbf{T} ,  \mathbf{Q} , \mathbf{q})$ from (\ref{eq:downlinkSINR_covar}) (under the small angle approximation) and using the matrix~identities
\begin{equation*}
\begin{split}
& \tr \left( \mathbf{A} \diag(\mathbf{B}) \right) = \tr \left( \mathbf{B} \diag(\mathbf{A}) \right),\\
& \tr \left( \mathbf{ABC} \right) = \tr \left( \mathbf{BCA} \right) = \tr \left( \mathbf{CAB} \right),
\end{split}
\end{equation*}
the $K$ DL SQINR constraints are given by (\ref{eq:downlinkSINR_covar2}).

The DL SQINR formulation $\gamma_k^{\text{DL}}(\mathbf{T} ,  \mathbf{Q} , \mathbf{q})$ in Section \ref{sec:DL} given by (\ref{eq:downlinkSINR_covar}) is equivalently given by the simplified expression $\gamma_k^{\text{DL}}(\mathbf{T} , \mathbf{q})$ on the right hand side (RHS) of (\ref{eq:downlinkSINR_covar2}) where the explicit dependance on the per-antenna power allocation matrix $\mathbf{Q}$ vanishes under the choice of $\mathbf{Q} = \diag \left(\sum_{i=1}^K q_i {\mathbf{t}}_i {\mathbf{t}}_i^\Her \right)^{\frac{1}{2}}$. Next, we define the $K \times K$ diagonal SQINR matrix $\mathbf{D}({\mathbf{T}})$
\begin{equation}
  \mathbf{D}({\mathbf{T}}) = \begin{bmatrix}
     \gamma_1/( {\mathbf{t}}_1^\Her \mathbf{R}_1 {\mathbf{t}}_1) & \dots & 0 \\      
     \vdots & \ddots  &\vdots\\    
     0 &  \dots & \gamma_K/( {\mathbf{t}}_K^\Her \mathbf{R}_K {\mathbf{t}}_K)
      \end{bmatrix}.
  \label{eq:DMatrix}
\end{equation} 
With $ \widehat{\mathbf{R}}_k = \left( \frac{\pi}{2} - 1 \right) \text{diag} \left(\mathbf{R}_k \right)$ and $ \widetilde{\mathbf{R}}_k = \mathbf{R}_k + \widehat{\mathbf{R}}_k $, we also define the $K \times K$ coupling matrix $\mathbf{\Psi}({\mathbf{T}})$ as
\begin{equation}
  \mathbf{\Psi}({\mathbf{T}}) = \begin{bmatrix}
     {\mathbf{t}}_1^\Her \widehat{\mathbf{R}}_1  {\mathbf{t}}_1 & {\mathbf{t}}_2^\Her \widetilde{\mathbf{R}}_1  {\mathbf{t}}_2  & \dots & {\mathbf{t}}_K^\Her \widetilde{\mathbf{R}}_1  {\mathbf{t}}_K \\ 
     
    {\mathbf{t}}_1^\Her \widetilde{\mathbf{R}}_2 {\mathbf{t}}_1 & {\mathbf{t}}_2^\Her \widehat{\mathbf{R}}_2  {\mathbf{t}}_2  & \dots & {\mathbf{t}}_K^\Her \widetilde{\mathbf{R}}_2 {\mathbf{t}}_K\\ 
    
    \vdots & \ddots & \ddots &\ddots\\ 
    
     {\mathbf{t}}_1^\Her \widetilde{\mathbf{R}}_K {\mathbf{t}}_1 & {\mathbf{t}}_2^\Her \widetilde{\mathbf{R}}_K {\mathbf{t}}_2  & \dots & {\mathbf{t}}_K^\Her \widehat{\mathbf{R}}_K  {\mathbf{t}}_K
      \end{bmatrix}.
  \label{eq:PsiMatrix}
\end{equation}
Using (\ref{eq:DMatrix}) and (\ref{eq:PsiMatrix}), the $K$ equations in (\ref{eq:downlinkSINR_covar2}) can be rearranged in matrix form as
\begin{equation} \label{eq:downlinkSINR_covar3}
\mathbf{q} = \mathbf{D}({\mathbf{T}}) \mathbf{\Psi}({\mathbf{T}}) \mathbf{q} + \frac{\pi \sigma^2}{2} \mathbf{D}({\mathbf{T}}) \mathbf{1}_K. \end{equation}
Using (\ref{eq:downlinkSINR_covar3}), the DL power allocation vector $\mathbf{q}$ can be written~as
\begin{equation} \label{eq:downlinkSINR_covar4}
\mathbf{q} =  \frac{\pi \sigma^2}{2} \left(  \mathbf{I}_K -  \mathbf{D}({\mathbf{T}}) \mathbf{\Psi}({\mathbf{T}}) \right)^{-1} \mathbf{D}({\mathbf{T}}) \mathbf{1}_K. \end{equation}
The existence of the matrix inverse in (\ref{eq:downlinkSINR_covar4}) is formally established in Lemma \ref{lemma:lemma1} and Lemma \ref{lemma:lemma2}.

\begin{lemma} \label{lemma:lemma1}
For any \emph{feasible} target DL SQINR set $\{\gamma_k\}$ with $k \in \{1, \dots K\}$, $\lambda_\maxx \left(\mathbf{D}({\mathbf{T}}) \mathbf{\Psi}({\mathbf{T}})\right) < 1$.
\end{lemma}

\begin{proof}
See Appendix \ref{sec:A}.
\end{proof}

\begin{lemma} \label{lemma:lemma2}
If $\lambda_\maxx \left(\mathbf{D}({\mathbf{T}}) \mathbf{\Psi}({\mathbf{T}})\right) < 1$, the matrix $\left(  \mathbf{I}_K -  \mathbf{D}({\mathbf{T}}) \mathbf{\Psi}({\mathbf{T}}) \right)$ is invertible.
\end{lemma}

\begin{proof}
See Appendix \ref{sec:B}.
\end{proof}\\
Lemma \ref{lemma:lemma1} and Lemma \ref{lemma:lemma2} establish that for a feasible target DL SQINR set $\{\gamma_k\}$, the choice of DL power allocation vector $\mathbf{q}$ in (\ref{eq:downlinkSINR_covar4}) achieves that for a given beamformer matrix $\mathbf{T}$.

%%%%%%%%%%%%%%%%%%%%
\subsection{Uplink SQINR}  \label{sec:dualityUL}
Now we show that the same set of SQINRs $\{\gamma_k\}$ for $k \in \{1, \dots K\}$ can be achieved in the UL by appropriately choosing the UL power allocation vector $\mathbf{p}$ and relating the UL combiners and DL beamformers. Defining ${\mathbf{t}}_k = \mathbf{A}_\text{u} \mathbf{u}_k / \| \mathbf{A}_\text{u} \mathbf{u}_k \|_2$ and observing that $\mathbf{u}_k^\Her \mathbf{I} \mathbf{u}_k$ equals $\mathbf{u}_k^\Her \mathbf{A}_\text{u}^\Her (\mathbf{A}_\text{u}^\Her)^{-1} \mathbf{A}_\text{u}^{-1} \mathbf{A}_\text{u} \mathbf{u}_k$, we equate the target UL SQINR $\gamma_k$ to the achieved UL SQINR $\gamma_k^{\text{UL}}(\mathbf{u}_k ,  \mathbf{p})$ in (\ref{eq:uplinkSINR_covar}) under the small angle approximation as
\begin{equation} \label{eq:uplinkSINR_covar2}
\gamma_k = \frac{p_k \mathbf{t}_k^\Her \mathbf{R}_k \mathbf{t}_k}{  \sum_{ \substack{i = 1\\ i \neq k}}^K p_i \mathbf{t}_k^\Her \mathbf{R}_i \mathbf{t}_k + \sigma^2 \mathbf{t}_k^\Her \mathbf{t}_k + \left( 1 - \frac{2}{\pi} \right) \mathbf{t}_k^\Her \mathbf{A}_\text{u}^{-2} \mathbf{t}_k }.
\end{equation}
By using the definition of $ \mathbf{A}_\text{u} = \sqrt{ \frac{2}{\pi} } \diag \left( \sum_{k = 1}^K p_k  \mathbf{R}_k + \sigma^2 \mathbf{I} \right)^{ -\frac{1}{2} }$ from Section \ref{sec:UL} and noting that $ \mathbf{t}_k^\Her \mathbf{t}_k = 1$, the $K$ UL SQINR constraints can be simplified to (\ref{eq:uplinkSINR_covar4}) given at the top of the next page.

%%%%%%
%%\addtocounter{equation}{1}
%%\setcounter{storeeqcounter}{\value{equation}}
%%
%%\begin{figure*}[t]
%%\normalsize
%%\setcounter{tempeqcounter}{\value{equation}} % temp store of current value
%%\begin{IEEEeqnarray}{rCl}
%%\setcounter{equation}{\value{storeeqcounter}} % number of this equation
%%\gamma_k = \frac{p_k \mathbf{t}_k^\Her \mathbf{R}_k \mathbf{t}_k}{  \sum_{ \substack{i = 1\\ i \neq k}}^K p_i \mathbf{t}_k^\Her \mathbf{R}_i \mathbf{t}_k + \sigma^2 \mathbf{t}_k^\Her \mathbf{t}_k + \left( \frac{\pi}{2} - 1 \right) \tr \left( \sum_{i = 1}^K p_i \mathbf{t}_k^\Her \diag( \mathbf{R}_i ) \mathbf{t}_k + \sigma^2 \mathbf{t}_k^\Her \mathbf{t}_k \right) }.
%%\label{eq:uplinkSINR_covar3}
%%\end{IEEEeqnarray}
%%\setcounter{equation}{\value{tempeqcounter}} % restore correct value
%%\end{figure*}

\addtocounter{equation}{1}
\setcounter{storeeqcounter}{\value{equation}}

\begin{figure*}[t]
\vspace{-0.0cm}
\normalsize
\setcounter{tempeqcounter}{\value{equation}} % temp store of current value
\begin{IEEEeqnarray}{rCl}
\setcounter{equation}{\value{storeeqcounter}} % number of this equation
\gamma_k = \frac{p_k \mathbf{t}_k^\Her \mathbf{R}_k \mathbf{t}_k}{  \sum_{ \substack{i = 1\\ i \neq k}}^K p_i \mathbf{t}_k^\Her \mathbf{R}_i \mathbf{t}_k + \frac{\pi}{2} \sigma^2 + \left( \frac{\pi}{2} - 1 \right) \tr \left( \sum_{i = 1}^K p_i \mathbf{t}_k^\Her \diag( \mathbf{R}_i ) \mathbf{t}_k \right) }.
\label{eq:uplinkSINR_covar4}
\end{IEEEeqnarray}
\hrulefill
\setcounter{equation}{\value{tempeqcounter}} % restore correct value
\end{figure*}

%%%%%
%By using , the $K$ equations in (\ref{eq:uplinkSINR_covar3}) can be further simplified to 

Under these simplifications, the UL SQINR formulation $\gamma_k^{\text{UL}}(\mathbf{u}_k ,  \mathbf{p})$ in Section \ref{sec:UL} given by (\ref{eq:uplinkSINR_covar}) is equivalently given by the expression $\gamma_k^{\text{UL}}(\mathbf{t}_k ,  \mathbf{p})$ on the right hand side (RHS) of (\ref{eq:uplinkSINR_covar4}) where the dependance on $\mathbf{u}_k$ has been cast in terms of $\mathbf{t}_k$. Using (\ref{eq:DMatrix}) and (\ref{eq:PsiMatrix}), the $K$ equations in (\ref{eq:uplinkSINR_covar4}) can be rearranged in matrix form as
\begin{equation} \label{eq:uplinkSINR_covar5}
\mathbf{p} = \mathbf{D}({\mathbf{T}}) \mathbf{\Psi}^\T({\mathbf{T}}) \mathbf{p} + \frac{\pi \sigma^2}{2} \mathbf{D}({\mathbf{T}}) \mathbf{1}_K.
\end{equation}
Using (\ref{eq:uplinkSINR_covar5}), the UL power allocation vector $\mathbf{p}$ (which achieves the same target SQINR set $\{\gamma_k\}$ for $k \in \{1 \dots K\}$ as the DL) can be written as
\begin{equation} \label{eq:uplinkSINR_covar6}
\mathbf{p} = \frac{\pi \sigma^2}{2} \left( \mathbf{I}_K - \mathbf{D}({\mathbf{T}}) \mathbf{\Psi}^\T({\mathbf{T}}) \right)^{-1} \mathbf{D}({\mathbf{T}}) \mathbf{1}_K.
\end{equation}
The existence of the matrix inverse in (\ref{eq:uplinkSINR_covar6}) is formally established in Lemma \ref{lemma:lemma12} and Lemma \ref{lemma:lemma22}.

\begin{lemma} \label{lemma:lemma12}
For any \emph{feasible} target UL SQINR set $\{\gamma_k\}$ with $k \in \{1, \dots K\}$, $\lambda_\maxx \left(\mathbf{D}({\mathbf{T}}) \mathbf{\Psi}^\T({\mathbf{T}})\right) < 1$.
\end{lemma}

\begin{proof}
The proof follows the proof of Lemma \ref{lemma:lemma1} in Appendix \ref{sec:A}.
\end{proof}

\begin{lemma} \label{lemma:lemma22}
If $\lambda_\maxx \left(\mathbf{D}({\mathbf{T}}) \mathbf{\Psi}^\T({\mathbf{T}})\right) < 1$, the matrix $\left(  \mathbf{I}_K -  \mathbf{D}({\mathbf{T}}) \mathbf{\Psi}^\T({\mathbf{T}}) \right)$ is invertible.
\end{lemma}

\begin{proof}
The proof follows the proof of Lemma \ref{lemma:lemma2} in Appendix \ref{sec:B}.
\end{proof}\\
Ignoring the scalar factor $\frac{\pi \sigma^2}{2}$, the total UL power allocation is given by
\begin{equation} \label{eq:duality}
\begin{split}
\|\mathbf{p}\|_1 &=   \mathbf{1}_K^\T \left( \mathbf{I}_K - \mathbf{D}({\mathbf{T}}) \mathbf{\Psi}^\T({\mathbf{T}}) \right)^{-1} \mathbf{D}({\mathbf{T}}) \mathbf{1}_K \\
& \overset{(a)}{=} \mathbf{1}_K^\T \mathbf{D}({\mathbf{T}}) \left( \mathbf{I}_K - \mathbf{\Psi}^\T({\mathbf{T}}) \mathbf{D}({\mathbf{T}}) \right)^{-1}  \mathbf{1}_K \\
& \overset{(b)}{=} \mathbf{1}_K^\T \mathbf{D}^\T({\mathbf{T}}) \left( \left( \mathbf{I}_K - \mathbf{D}({\mathbf{T}}) \mathbf{\Psi}({\mathbf{T}}) \right)^{-1} \right)^\T  \mathbf{1}_K \\
& \overset{(c)}{=} \left( \left( \mathbf{I}_K - \mathbf{D}({\mathbf{T}}) \mathbf{\Psi}({\mathbf{T}}) \right)^{-1} \mathbf{D}({\mathbf{T}}) \mathbf{1}_K \right)^\T  \mathbf{1}_K \\
& \overset{(d)}{=} \mathbf{q}^\T \mathbf{1}_K = \|\mathbf{q}\|_1,
\end{split}
\end{equation}
where (a) follows from the push-through identity, (b) follows from the diagonal structure of $\mathbf{D}({\mathbf{T}})$, (c) follows from $(\mathbf{AB})^\T = \mathbf{B}^\T \mathbf{A}^\T$, and (d) follows from (\ref{eq:downlinkSINR_covar4}). It can be seen from (\ref{eq:duality}) that the same amount of total power is needed in the DL and UL to achieve the same SQINR values for the $K$ users. This formally establishes the UL-DL duality principle stated in Theorem \ref{theorem:theorem1} for systems with 1-bit ADCs/DACs under the small angle approximation. In Section \ref{sec:solution}, we will use this UL-DL duality for solving the MU-DL-BF problem under 1-bit DAC constraints. 

\section{UL-DL duality based  proposed solution} \label{sec:solution}
In this section, we first introduce the MU-DL-BF problem with individual SQINR constraints considered in this paper. We then describe the optimal power allocation strategy for the DL and UL stages for a fixed choice of the beamforming/combining matrices. We conclude this section by providing the details of the alternating minimization algorithm for joint optimization of the linear beamformers and power allocation for each user.

%%%%%%%%%%%%%%%%%%
\vspace{-0.2cm}
\subsection{Problem formulation} \label{sec:problem}
Let $\gamma_k$ for $1 \leq k \leq K$ be the individual target SQINRs for each user and $P_\BS$ Watts be the total DL power budget. In this paper, we aim to maximize the minimum of the $K$ achieved to target SQINR ratios over all possible beamforming matrices and power allocation vectors. This formulation has not been considered before in the context of systems with 1-bit constraints and we believe that it provides additional flexibility as will be seen in Section \ref{sec:results}. The proposed formulation of the MU-DL-BF problem for K users, each with an individual SQINR constraint, can be stated as
\begin{equation} \label{eq:DL_SINR}
\begin{aligned}
R^\text{DL}_\text{opt}(P_\BS) = \max_{\mathbf{T,q}} \min_{1\leq k \leq K} \quad & \frac{\gamma_k^{\text{DL}}(\mathbf{T , q})}{\gamma_k} \\
\textrm{s.t.} \quad & \|\mathbf{q}\|_1 \leq P_\BS\\
& ||\mathbf{t}_k||_2 = 1, \quad 1\leq k \leq K. \\
\end{aligned}
\end{equation}
A closely related problem to (\ref{eq:DL_SINR}) is maximizing the minimum SQINR given by $\max_{\mathbf{T,q}} \min_{1\leq k \leq K} {\gamma_k^{\text{DL}}(\mathbf{T , q})}$. This, however, is just a special case of (\ref{eq:DL_SINR}) when the target SQINRs $\gamma_k$ are taken to be equal for all $K$ users.
We also consider an `easier' version of the problem (\ref{eq:DL_SINR}) where the minimum of the $K$ achieved to target SQINR ratios has to be maximized over all possible power allocation vectors for a fixed DL-BF matrix $\mathbf{T}^\star$. This power allocation problem is given by
\begin{equation} \label{eq:DL_SINR2}
\begin{aligned}
R^\text{DL}_\text{opt}(P_\BS , \mathbf{T}^\star) = \max_{\mathbf{q}} \min_{1\leq k \leq K} \quad & \frac{\gamma_k^{\text{DL}}(\mathbf{T^\star , q})}{\gamma_k} \\
\textrm{s.t.} \quad & \|\mathbf{q}\|_1 \leq P_\BS.\\
\end{aligned}
\end{equation}
The MU-UL-BF problem and power allocation problem for a fixed UL-BF matrix can be cast in the same manner.

%%%%%%%%%%%%%%%%%%
\vspace{-0.2cm}
\subsection{Optimal DL power allocation} \label{sec:DLpower}
The function $R^\text{DL}_\text{opt}(P_\BS , \mathbf{T}^\star)$ in (\ref{eq:DL_SINR2}) is strictly monotonically increasing in $P_\BS$. This is a consequence of the DL SQINR $\gamma_k^{\text{DL}}(\mathbf{T} , \mathbf{q})$ in (\ref{eq:downlinkSINR_covar2}) being a monotonically increasing function of $P_\BS$. The constant $\frac{\pi \sigma^2}{2}$ in (\ref{eq:downlinkSINR_covar2}) results in $\gamma_k^{\text{DL}}(\mathbf{T} , \alpha \mathbf{q}) > \gamma_k^{\text{DL}}(\mathbf{T} , \mathbf{q})$ for $\alpha > 1$. This allows us to characterize the maximizer of the power allocation problem in (\ref{eq:DL_SINR2}) by Lemma~\ref{lemma:lemma3}.

\begin{lemma} \label{lemma:lemma3}
Let $\mathbf{q}^\star$ $\left(\|\mathbf{q}^\star\|_1 = P_\BS \right)$ be the solution to the optimal DL power allocation problem in (\ref{eq:DL_SINR2}) for a fixed BF $\mathbf{T}^\star$. The optimizer $\mathbf{q}^\star$ results in \emph{balanced} achieved SQINR to target SQINR ratio for all $K$ users given by
\begin{equation} \label{eq:DL_SINR3}
\begin{aligned}
R^\text{DL}_\text{opt}(P_\BS , \mathbf{T}^\star) = & \frac{\gamma_k^{\text{DL}}(\mathbf{T^\star , q^\star})}{\gamma_k}, \quad {1\leq k \leq K}.
\end{aligned}
\end{equation}
\end{lemma}

\begin{proof}
See Appendix \ref{sec:C}.
\end{proof}\\
It can be seen from Lemma \ref{lemma:lemma3} that the optimal power allocation vector $\mathbf{q}^\star$ results in equal achieved SQINR to target SQINR ratios for all $K$ users. This allows us to characterize the feasibility of the target SQINR set $\{\gamma_k\}$ for $k \in \{1 \dots K\}$ using Corollary \ref{cor:cor1} which follows from Lemma \ref{lemma:lemma3}.

\begin{corollary} \label{cor:cor1}
The target SQINR set $\{\gamma_k\}$ for $k \in \{1 \dots K\}$ is achievable under a total power constraint of $P_\BS$ if and only if $R^\text{DL}_\text{opt}(P_\BS , \mathbf{T}^\star) \geq 1$.
\end{corollary}
A further consequence of Lemma \ref{lemma:lemma3} is that if $R^\text{DL}_\text{opt}(P_\BS , \mathbf{T}^\star) < 1$, then the target SQINRs $[\gamma_1 \dots \gamma_k \dots \gamma_K]$ are not met. In that setting, the achieved SQINRs are an \emph{equal} fractional multiple $(\leq 1)$ of the target SQINRs $[\gamma_1 \dots \gamma_k \dots \gamma_K]$ for all $K$ users.

Using the matrix definitions (\ref{eq:DMatrix}) and (\ref{eq:PsiMatrix}) and putting in the value of the $\gamma_k^{\text{DL}}(\mathbf{T^\star , q^\star})$ from (\ref{eq:downlinkSINR_covar2}), the $K$ equations in (\ref{eq:DL_SINR3}) can be written in matrix form as
\begin{equation} \label{eq:power1}
\mathbf{q}^\star \frac{1}{R^\text{DL}_\text{opt}(P_\BS , \mathbf{T}^\star)} = \mathbf{D}(\mathbf{T}^\star) \mathbf{\Psi}(\mathbf{T}^\star) \mathbf{q}^\star + \frac{\pi \sigma^2}{2} \mathbf{D}(\mathbf{T}^\star) \mathbf{1}_K.
\end{equation}
Multiplying both sides by $\mathbf{1}_K^\T$ and noting that $\mathbf{1}_K^\T \mathbf{q}^\star = P_\BS$
\begin{equation} \label{eq:power2}
\frac{1}{R^\text{DL}_\text{opt}(P_\BS , \mathbf{T}^\star)} = \frac{ \mathbf{1}_K^\T \mathbf{D}(\mathbf{T}^\star) \mathbf{\Psi}(\mathbf{T}^\star) \mathbf{q}^\star}{P_\BS} + \frac{\pi \sigma^2}{2 P_\BS} \mathbf{1}_K^\T \mathbf{D}(\mathbf{T}^\star) \mathbf{1}_K.
\end{equation}
Defining the extended DL power allocation vector $\mathbf{q}^\star_\text{ext} = \begin{bmatrix} \mathbf{q}^\star \quad 1 \end{bmatrix}^\T$ and the  non-negative extended DL coupling matrix
\begin{equation}
  \mathbf{\Upsilon}(\mathbf{T}^\star , P_\BS) = \begin{bmatrix}
    \mathbf{D(T^\star)} \mathbf{\Psi(T^\star)} &  \frac{\pi \sigma^2}{2}\mathbf{D(T^\star)} \mathbf{1}_K  \\ 
        
   \frac{\mathbf{1}_K^\T}{P_\BS}\mathbf{D(T^\star)} \mathbf{\Psi(T^\star)} &  \frac{\pi \sigma^2}{2 P_\BS} \mathbf{1}_K^\T  \mathbf{D(T^\star)}\mathbf{1}_K
      \end{bmatrix},
  \label{eq:psiMatrix}
\end{equation}
it can be seen that (\ref{eq:power1}) and (\ref{eq:power2}) form an eigensystem with
\begin{equation} \label{eq:psiMatrix2}
\mathbf{\Upsilon}(\mathbf{T}^\star , P_\BS) \mathbf{q}^\star_\text{ext} = \frac{1}{R^\text{DL}_\text{opt}(P_\BS , \mathbf{T}^\star)} \mathbf{q}^\star_\text{ext}.
\end{equation}
It can be observed that the achieved SQINR to target SQINR balance value, $R^\text{DL}_\text{opt}(P_\BS , \mathbf{T}^\star)$, equals the reciprocal of the eigenvalue of the extended coupling matrix $\mathbf{\Upsilon}(\mathbf{T}^\star , P_\BS)$. Not all eigenvalues of $\mathbf{\Upsilon}(\mathbf{T}^\star , P_\BS)$ are meaningful. Particularly, $R^\text{DL}_\text{opt}(P_\BS , \mathbf{T}^\star) > 0$ and $\mathbf{q}^\star_\text{ext} > 0$ must be fulfilled. It is known from Perron-Frobenius theory that for a non-negative matrix $\mathbf{B}$, there exists a non-negative vector $\mathbf{b} \geq 0$ such that $\mathbf{Bb} = \lambda_\maxx(\mathbf{B}) \mathbf{b}$ \cite{MUDL}. Furthermore, it has been proven that for a non-negative matrix $\mathbf{\Upsilon}(\mathbf{T}^\star , P_\BS)$ with the structure in (\ref{eq:psiMatrix}), the maximal eigenvalue (and its associated eigenvector) are strictly positive and no other eigenvalue fulfills the positivity criterion \cite{MUDL}. Hence the SQINR balancing solution to the DL power allocation problem in (\ref{eq:DL_SINR2}) is given by
\begin{equation} \label{eq:psiMatrix3}
R^\text{DL}_\text{opt}(P_\BS , \mathbf{T}^\star) = \frac{1}{ \lambda_\maxx \left( \mathbf{\Upsilon}(\mathbf{T}^\star , P_\BS) \right) }.
\end{equation}
And the optimal DL power allocation vector $\mathbf{q}^\star$ is given by the first $K$ entries of the dominant eigenvector of $\mathbf{\Upsilon}(\mathbf{T}^\star , P_\BS)$ scaled such that the last entry is equal to one.

%that $\mathbf{q}^\star_\text{ext}$ is an eigenvector of $\mathbf{\Upsilon}(\mathbf{T}^\star , P_\BS)$ with the eigenvalue $1/R^\text{DL}_\text{opt}(P_\BS , \mathbf{T}^\star)$. 

%%%%%%%%%%%%%%%%%%
\vspace{-0.2cm}
\subsection{Optimal UL power allocation} \label{sec:ULpower}
Now we consider the UL scenario for the same total power constraint $P_\BS$, the same target SQINR set $\{\gamma_k\}$ for $k \in \{1 \dots K\}$, and the same combining (beamforming for DL) matrix $\mathbf{T}^\star$. Following the development in Section \ref{sec:DLpower}, we define the extended UL power allocation vector $\mathbf{p}^\star_\ext = \begin{bmatrix} \mathbf{p}^\star \quad 1 \end{bmatrix}^\T$ and the non-negative extended UL coupling matrix
\begin{equation}
  \mathbf{\Lambda}(\mathbf{T}^\star , P_\BS) = \begin{bmatrix}
    \mathbf{D(T^\star)} \mathbf{\Psi}^\T(\mathbf{T}^\star) &  \frac{\pi \sigma^2}{2}\mathbf{D(T^\star)} \mathbf{1}_K  \\ 
        
   \frac{\mathbf{1}_K^\T}{P_\BS}\mathbf{D(T^\star)} \mathbf{\Psi}^\T (\mathbf{T}^\star) &  \frac{\pi \sigma^2}{2 P_\BS} \mathbf{1}_K^\T  \mathbf{D(T^\star)}\mathbf{1}_K
      \end{bmatrix}.
  \label{eq:gammaMatrix}
\end{equation}
Like their DL counterparts, it can be shown that the $\mathbf{p}^\star_\text{ext}$ and $\mathbf{\Lambda}(\mathbf{T}^\star , P_\BS)$ form an eigensystem with
\begin{equation} \label{eq:gammaMatrix3}
\mathbf{\Lambda}(\mathbf{T}^\star , P_\BS) \mathbf{p}^\star_\text{ext} = \frac{1}{R^\text{UL}_\text{opt}(P_\BS , \mathbf{T}^\star)} \mathbf{p}^\star_\text{ext}.
\end{equation}
The SQINR balancing solution to the UL power allocation problem is given by
\begin{equation} \label{eq:gammaMatrix2}
R^\text{UL}_\text{opt}(P_\BS , \mathbf{T}^\star) = \frac{1}{ \lambda_\maxx \left( \mathbf{\Lambda}(\mathbf{T}^\star , P_\BS) \right) }.
\end{equation}
The optimal UL power allocation is given by the first $K$ entries of the dominant eigenvector of $ \mathbf{\Lambda}(\mathbf{T}^\star , P_\BS)$ scaled such that the $(K+1)^\thh$ entry equals one. The UL and DL balanced SQINR ratios are related by Lemma \ref{lemma:lemma4}.

\begin{lemma} \label{lemma:lemma4}
The UL and DL achieved SQINR to target SQINR ratios are equal, i.e.\ $R^\text{UL}_\text{opt}(P_\BS , \mathbf{T}^\star) = R^\text{DL}_\text{opt}(P_\BS , \mathbf{T}^\star)$.
\end{lemma}

\begin{proof}
This follows from Theorem \ref{theorem:theorem1}.
\end{proof}\\
It follows from UL-DL duality established in Section \ref{sec:duality} that the same achieved SQINR to target SQINR ratio is achieved in both the DL and UL, albeit for different power allocation vectors. We now exploit this property to recast the MU-DL-BF problem in terms of the `easier-to-solve' MU-UL-BF problem.

%%%%%%%%%%%%%%%%%%%%
\vspace{-0.2cm}
\subsection{Joint power allocation and precoder design} \label{sec:joint} 

Having established that the optimum solution to the DL power allocation problem balances the ratios $\frac{\gamma_k^{\text{DL}}(\mathbf{T^\star , q^\star})}{\gamma_k}$ at a common level $R^\text{DL}_\text{opt}(P_\BS , \mathbf{T}^\star)$, we now maximize $R^\text{DL}_\text{opt}(P_\BS , \mathbf{T}^\star)$ over all unit norm beamformers to solve the joint power allocation and beamformer optimization problem (\ref{eq:DL_SINR}). Since $1/R^\text{DL}(P_\BS , \mathbf{T}^\star)$ is the dominant eigenvalue associated with the extended coupling matrix $\mathbf{\Upsilon}(\mathbf{T}^\star , P_\BS)$ for a fixed BF $\mathbf{T}^\star$, the joint power and beamformer optimization problem (\ref{eq:DL_SINR}) can be equivalently stated as
\begin{equation} \label{eq:DL_SINR6}
R^\text{DL}_\text{opt}(P_\BS) = \frac{1}{\min_{\mathbf{T}} \lambda_{\max}\left( \mathbf{\Upsilon}(\mathbf{T}, P_\BS) \right) }.
\end{equation} 
By Lemma \ref{lemma:lemma4}, the optimum solution to (\ref{eq:DL_SINR}) can also be achieved by the equivalent UL formulation
\begin{equation} \label{eq:DL_SINR7}
R^\text{DL}_\text{opt}(P_\BS) = \frac{1}{\min_{\mathbf{T}} \lambda_{\max}\left( \mathbf{\Lambda}(\mathbf{T}, P_\BS) \right) }.
\end{equation} 
By the Perron-Frobenius theorem \cite{MUDL},
\begin{equation} \label{eq:DL_SINR8}
\min_{\mathbf{T}} \lambda_{\max}\left( \mathbf{\Lambda}(\mathbf{T}, P_\BS) \right) =  \min_{\mathbf{T}} \max_{\mathbf{x}>0} \min_{\mathbf{y}>0} \frac{ \mathbf{x}^\T \mathbf{\Lambda}(\mathbf{T}, P_\BS) \mathbf{y}  }{ \mathbf{x}^\T \mathbf{y} } .
\end{equation} 
Next, we define an intermediate cost function
\begin{equation} \label{eq:DL_SINR9}
 \widehat{\lambda} \left( \mathbf{T}, P_\BS, \mathbf{p}_\ext \right) = \max_{\mathbf{x}>0}  \frac{ \mathbf{x}^\T \mathbf{\Lambda}(\mathbf{T}, P_\BS) \mathbf{p}_\ext  }{ \mathbf{x}^\T \mathbf{p}_\ext },
\end{equation} 
that lends the equivalent problem formulation

\begin{equation} \label{eq:DL_SINR10}
\left( R^\text{DL}_\text{opt}(P_\BS) \right)^{-1} = \min_{\mathbf{T}} \min_{\mathbf{p}_\ext > 0} \widehat{\lambda} \left( \mathbf{T}, P_\BS, \mathbf{p}_\ext \right).
\end{equation} 
The variable $\mathbf{x}$ in (\ref{eq:DL_SINR9}) is an auxiliary optimization variable and has no physical meaning. Following closely the development for $\infty-$resolution ADCs/DACs MU-DL-BF problem \cite{MUDL}, we propose an alternating minimization solution to the joint power allocation and precoder design problem (\ref{eq:DL_SINR}) where either the extended UL power allocation vector $\mathbf{p}_\text{ext}$ or the beamforming matrix $\mathbf{T}$ is held fixed while the other is optimized. The representation in (\ref{eq:DL_SINR10}) enables the proposed algorithmic solution based on alternating optimization.

%%%%%%
\subsubsection{Power allocation step}
For a fixed beamforming matrix $\mathbf{T}$, the function $\widehat{\lambda} \left( \mathbf{T}, P_\BS, \mathbf{p}_\ext \right)$ is minimized by the power allocation vector which solved the eigenvalue problem (\ref{eq:gammaMatrix3}). This follows from the optimal power allocation procedure detailed in Section \ref{sec:DLpower} and \ref{sec:ULpower}. It can also be shown by multiplying both sides of (\ref{eq:gammaMatrix3}) by $\mathbf{x}^\T$ and dividing by $\mathbf{x}^\T \mathbf{p}_\ext^\star$.
\begin{equation} \label{eq:gammaMatrix4}
\begin{split}
\frac{ \mathbf{x}^\T \mathbf{\Lambda}(\mathbf{T} , P_\BS) \mathbf{p}^\star_\text{ext} } { \mathbf{x}^\T \mathbf{p}_\ext^\star } &= \frac{1}{R^\text{UL}_\text{opt}(P_\BS , \mathbf{T}^\star)} \frac{ \mathbf{x}^\T \mathbf{p}_\ext^\star }{ \mathbf{x}^\T \mathbf{p}_\ext^\star }\\
&= \lambda_{\max}\left( \mathbf{\Lambda}(\mathbf{T}, P_\BS) \right) \\
&= \min_{\mathbf{p}_\ext > 0} \widehat{\lambda} \left( \mathbf{T}, P_\BS, \mathbf{p}_\ext \right).
\end{split}
\end{equation}

%%%%%%
\subsubsection{Beamformer optimization step}
Now we focus on optimizing the beamforming matrix $\mathbf{T}$ for a fixed power allocation vector $\mathbf{p}$ $\left(\text{with } \mathbf{p}_\ext = \begin{bmatrix} \mathbf{p} \quad 1 \end{bmatrix}^\T \right)$ given by the problem
\begin{equation} \label{eq:P1}
\mathbf{T}^\star = \argmin_{\mathbf{T}} \widehat{\lambda} \left( \mathbf{T}, P_\BS, \mathbf{p}_\ext \right).
\end{equation}
Lemma \ref{lemma:lemma5} (proved in Appendix \ref{sec:D}) provides a useful intermediate step in this direction. 

\begin{lemma} \label{lemma:lemma5}
The cost function $\widehat{\lambda} \left( \mathbf{T}, P_\BS, \mathbf{p}_\ext \right)$ can equivalently be written as
\begin{equation}\label{eq:P2}
\max_{\mathbf{x}>0}  \frac{ \mathbf{x}^\T \mathbf{\Lambda}(\mathbf{T}, P_\BS) \mathbf{p}_\ext  }{ \mathbf{x}^\T \mathbf{p}_\ext } = \max_{1 \leq k \leq K} \frac{\gamma_k}{\gamma_k^{\text{UL}}\left(  \mathbf{t}_k ,  \mathbf{p} \right) }.
\end{equation}
\end{lemma}
Corollary \ref{cor:cor2} follows from Lemma \ref{lemma:lemma5}.

\begin{corollary} \label{cor:cor2}
The solution to the problem (\ref{eq:P1}) is given by independent maximization of the $K$ UL SQINRs $\gamma_k^{\text{UL}}\left(  \mathbf{t}_k ,  \mathbf{p} \right)$.
\end{corollary}
This allows us to decouple the joint optimization problem (\ref{eq:P1}) into $K$ decoupled problems. Using the definition of the UL SQINR (\ref{eq:uplinkSINR_covar4}), the beamformer $\mathbf{t}_k^\star$ maximizing the UL SQINR $\gamma_k^{\text{UL}}\left(  \mathbf{t}_k ,  \mathbf{p} \right)$ is given by
\begin{equation} \label{eq:P3}
\mathbf{t}_k^\star = \argmax_{\mathbf{t}_k} \frac{p_k \mathbf{t}_k^\Her \mathbf{R}_k \mathbf{t}_k}{\mathbf{t}_k^\Her \mathbf{Q}_k(\mathbf{p}) \mathbf{t}_k}, \quad  \textrm{s.t.} \| \mathbf{t}_k \|_2 = 1,
\end{equation} 
where
\begin{equation} \label{eq:DL_SINR20}
\mathbf{Q}_k( \mathbf{p} ) =  \sum_{ \substack{i = 1\\ i \neq k}}^K p_i \mathbf{R}_i + \left( \frac{\pi}{2} - 1 \right) \sum_{i=1}^K p_i \diag(\mathbf{R}_i)  + \frac{\pi}{2} \sigma^2 \mathbf{I}.
\end{equation} 
Since the matrices $\mathbf{R}_k$ and $\mathbf{Q}_k$ are hermitian, the solution to (\ref{eq:P3}) is given by the dominant generalized eigenvector of the matrix pair $( \mathbf{R}_k , \mathbf{Q}_k )$ for $1 \leq k \leq K$ \cite{MUDL}. 

\begin{algorithm}[t]
\caption{Alternating minimization solution to (\ref{eq:DL_SINR}) }\label{alg:alg1}
 1) \text{Initialize:} $n = 0, \mathbf{p}^{(0)} = [0, \dots , 0]^\T , P_\BS , \epsilon$ \vspace{.1cm}  \\
 2) \textbf{while} $\lambda_{\max}^{(n-1)} - \lambda_{\max}^{(n)} \geq \epsilon$ \vspace{.05cm}    \\
 3) \hspace{0cm} $n = n+1$ \vspace{.05cm}    \\
 4) \hspace{0cm}  $\forall_k \text{ } {\mathbf{t}}_k^{\star(n)} = \mathbf{v}_{\max}\left( \mathbf{R}_k,  \mathbf{Q}_k(\mathbf{p}^{(n-1)})\right)(\mathbf{v}_{\max} \triangleq$ eigenvector)   \vspace{0cm}    \\
 5) \hspace{0cm}  $\forall_k \text{ } {\mathbf{t}}_k^{\star(n)} = \mathbf{t}_k^{\star(n)} / \| \mathbf{t}_k^{\star(n)} \|_2 $ \hspace{2.1cm} (normalization)   \vspace{0cm}    \\
 6) \hspace{0cm}  $\mathbf{\Lambda}({\mathbf{T}}^{\star(n)}, P_\BS) \mathbf{p}_\text{ext}^{\star(n)} =  \lambda_{\max}^{(n)} \mathbf{p}_\text{ext}^{\star(n)}$  \hspace{0.1cm} (UL power allocation) \vspace{0cm}\\
 7) \hspace{0cm}  $\mathbf{p}^{\star(n)} = \mathbf{p}_\text{ext}^{\star(n)}[1,\dots,K] / \mathbf{p}_\text{ext}^{\star(n)}[K+1]$ \hspace{0.4cm}(normalization)\vspace{.1cm}\\
8) \textbf{end} \vspace{.1cm}\\
9) $\mathbf{\Upsilon}({\mathbf{T}}^{\star(n)}, P_\BS) \mathbf{q}_\text{ext}^{\star} =  \lambda_{\max}^{(n)} \mathbf{q}_\text{ext}^{\star} $   \hspace{0.8cm} (DL power allocation) \vspace{.05cm} \\
10) $\mathbf{q}^{\star(n)} = \mathbf{q}_\text{ext}^{\star(n)}[1,\dots,K] / \mathbf{q}_\text{ext}^{\star(n)}[K+1]$ \hspace{0.5cm}(normalization)
\end{algorithm}

The joint beamformer and power allocation algorithm is run by alternating between the beamformer optimization step for a fixed UL power allocation vector and the UL power optimization step for a fixed beamforming matrix. These two steps are repeated till $ \lambda^{(n-1)}_{\max}\left( \mathbf{\Lambda}(\mathbf{T}, P_\BS) \right) -  \lambda^{(n)}_{\max}\left( \mathbf{\Lambda}(\mathbf{T}, P_\BS) \right) < \epsilon$. Here the superscript $(\cdot)^{(n)}$ denotes the iteration index and $\epsilon$ is a predefined constant which controls when to stop the optimization procedure. Finally, the DL power allocation vector $\mathbf{q}^\star$ is calculated using the precoder matrix $\mathbf{T}^\star$ obtained in the final iteration. The proposed solution is summarized in Algorithm~\ref{alg:alg1}.

%%%%%%%%%%%%%%%%%%%%
\subsection{Convergence} \label{sec:convergence} 
The DL power allocation and beamforming optimization problem under 1-bit DAC constraints is NP-hard. Different non-linear algorithms get to a computationally feasible solution by relaxing the non-convex 1-bit constraints in various manners \cite{hela2016,hela,studer2016,Studer,studer2017}. Similar to \cite{hela2016,hela,studer2016,Studer,studer2017}, our proposed solution in Algorithm~\ref{alg:alg1} has no global optimality guarantees. Our simulations, however, do indicate that the proposed algorithm typically converges within $2-5$ iterations. In this subsection, we show that the proposed alternating minimization procedure does indeed converge to some point in the solution space. Let $ \lambda_{\max}^{(n)}  =  \lambda_{\max}( \mathbf{\Lambda}({\mathbf{T}}^{\star(n)}, P_\BS))$ denote the reciprocal of the achieved SQINR to target SQINR value at the end of the $n^\thh$ iteration (step 6 in Algorithm \ref{alg:alg1}). It is clear from the precoder optimization step that the precoder matrix in the $(n+1)^\thh$ iteration, $\mathbf{T}^{\star(n+1)}$, minimizes the cost function  $\widehat{\lambda} \left( \mathbf{T}, P_\BS,  \mathbf{p}_\text{ext}^{\star(n)} \right)$ i.e.
\begin{equation} \label{eq:conv1}
\widehat{\lambda} \left( \mathbf{T}^{\star(n+1)}, P_\BS,  \mathbf{p}_\text{ext}^{\star(n)} \right) \leq \widehat{\lambda} \left( \mathbf{T}^{\star(n)}, P_\BS,  \mathbf{p}_\text{ext}^{\star(n)} \right) = \lambda_{\max}^{(n)}.
\end{equation}
From Perron-Frobenius theorem \cite{MUDL}, we know that 
\begin{equation} \label{eq:conv2}
\begin{split}
\lambda_{\max}^{(n+1)} &= \max_{\mathbf{x}>0} \min_{\mathbf{y}>0} \frac{ \mathbf{x}^\T \mathbf{\Lambda}(\mathbf{T}^{\star(n+1)}, P_\BS) \mathbf{y}  }{ \mathbf{x}^\T \mathbf{y} } \\
&\leq \max_{\mathbf{x}>0} \frac{ \mathbf{x}^\T \mathbf{\Lambda}(\mathbf{T}^{\star(n+1)}, P_\BS) \mathbf{p}_\text{ext}^{\star(n)}  }{ \mathbf{x}^\T \mathbf{p}_\text{ext}^{\star(n)} } \\
&= \widehat{\lambda} \left( \mathbf{T}^{\star(n+1)}, P_\BS,  \mathbf{p}_\text{ext}^{\star(n)} \right).
\end{split}
\end{equation}
Combining (\ref{eq:conv1}) and (\ref{eq:conv2}), it can be seen that
\begin{equation} \label{eq:conv3}
\lambda_{\max}^{(n+1)} \leq \widehat{\lambda} \left( \mathbf{T}^{\star(n+1)}, P_\BS,  \mathbf{p}_\text{ext}^{\star(n)} \right) \leq \lambda_{\max}^{(n)}.
\end{equation}
The non-negativity of the $\lambda_{\max}^{(n)}$ combined with the monotonic behavior in (\ref{eq:conv3}) implies the existence of a limiting value $\lambda_{\max}^{(\infty)}$. The degree to which the solution given by Algorithm \ref{alg:alg1} approaches this value can be controlled by varying the parameter $\epsilon$.

%%%%%%%%%%%%%%%%%%%%%%%%%%%%%%%%%%%%%%%%%%%%%%%%%%%%%%%%
\section{Optimized dithering by dummy users}\label{sec:dummy}
The small angle approximation in Section \ref{sec:uqn} was crucial in proving the UL-DL duality in Section \ref{sec:duality} and in the formulation of the proposed algorithm in Section \ref{sec:solution}. The small angle assumption basically says that the off-diagonal elements of the covariance matrix of the signal before quantization, $\mathbf{R}_{\mathbf{y}_\text{d}} $ and $\mathbf{R}_{\mathbf{y}_\text{u}} $, are small compared to the diagonal entries. This assumption, however, is realistic in DL scenarios only for a large number of active users. For a small number of active users, the terms involved in the off-diagonal entries of $\mathbf{R}_{\mathbf{y}_\text{d}} $ do not undergo enough averaging and can be close to the diagonal entries. Using (\ref{eq:upn4}), this makes the quantization noise correlated which lowers the achievable SQINR. Similarly for UL settings, the small angle approximation does not remain true when the SQINR is high or the number of active users is small. These realistic considerations motivate the addition of optimized dithering to the system to ensure that the quantization noise is uncorrelated. We define a metric $d( \mathbf{B} )$ (and call it the `diagonal metric') to quantify the degree to which a matrix $\mathbf{B}$ is diagonal as
\begin{equation} \label{eq:dummy1}
d( \mathbf{B} ) = \frac{ \| \diag(\mathbf{B}) \|_F }{ \| \mathbf{B} \|_F }.
\end{equation}
It is clear from (\ref{eq:dummy1}) that $d( \mathbf{B} ) = 1$ for a diagonal matrix and that $d( \mathbf{B} ) \approx 1$ for a matrix that has the majority of weight concentrated in the diagonal entries. Fig. \ref{fig:diag} illustrates the diagonal metric in blue (\tikz\draw [blue,thick] (0,0) -- (.5,0);) for the DL quantization noise matrix $\mathbf{R}_{\pmb{\eta}_\text{d}}$ vs the number of active users for precoders and power allocation vectors obtained through the proposed algorithm in Section \ref{sec:solution}. It is clear from Fig. \ref{fig:diag} that the uncorrelated quantization noise assumption is violated for scenarios with a small number of active users. These scenarios are where optimized dithering can help boost the achievable SQINR and improve performance. Similar observations for ZF based precoders were also made for settings with a small number of users \cite{amodh}.

\begin{figure}[t]
    	\begin{center}
    		\includegraphics[width=.48\textwidth,clip,keepaspectratio]{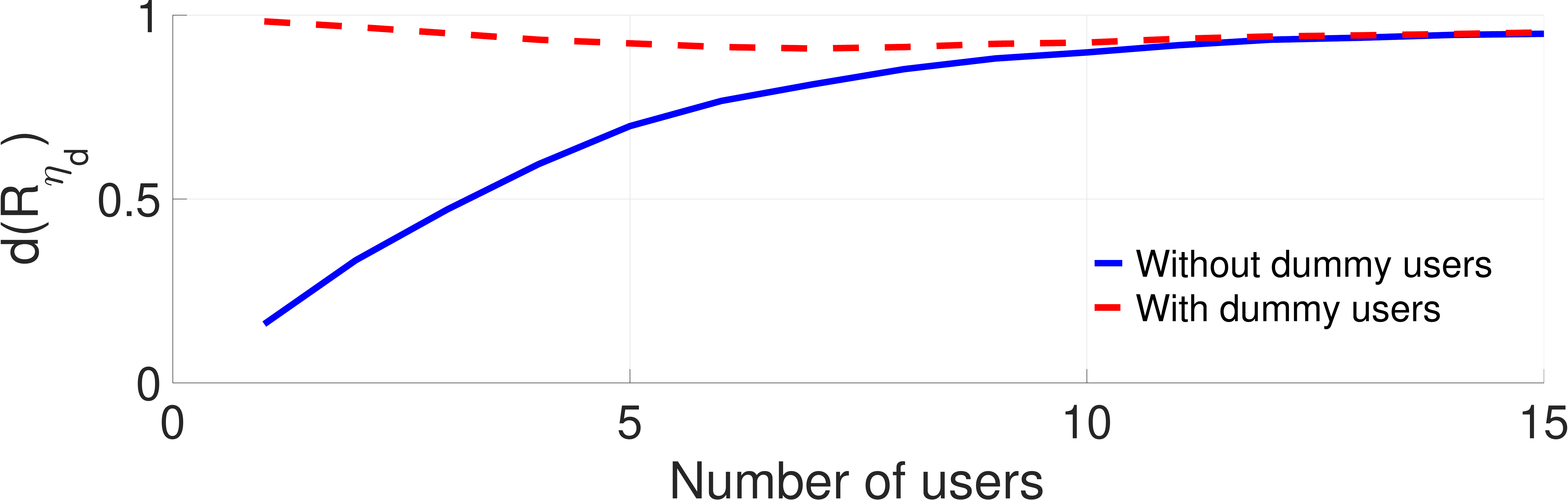}
    	\end{center}
\vspace{-0.2cm}
    	\caption{The diagonal metric $d( \mathbf{R}_{\pmb{\eta}_\text{d}} )$ with and without dummy users vs number of true system users. It can be seen that after adding the dummy users in the system, the uncorrelated quantization noise assumption is always satisfied irrespective of the number of true system users.}
    	 \label{fig:diag}
\vspace{-0.5cm}
\end{figure}

The main purpose of adding dither to the signal before quantization is to force the small angle approximation to be true by making the off-diagonal entries of $\mathbf{R}_{\mathbf{y}_\text{d}}$ and $\mathbf{R}_{\mathbf{y}_\text{u}}$ small compared to the diagonal entries. We add optimized dithering by adding $(N_\BS - K)$ dummy users to the system each with its own channel and appropriate SQINR constraint. The dummy users lie in the nullspace of the true users and thus pose minimal MUI while significantly lowering the quantization noise. The channel matrix of the $(N_\BS - K)$ dummy users, ${\mathbf{H}}_d \in \mathbb{C}^{(N_\BS - K) \times N_\text{BS}}$, is defined as
\begin{equation} \label{eq:dummy2}
{\mathbf{H}}_d = [{\mathbf{h}}_{K+1} \dots {\mathbf{h}}_{N_\BS}]^\T = \mathbb{N}({\mathbf{H}}).
\end{equation}
Similar to the true system users, the channel covariance matrices of the dummy users are defined as $\mathbf{R}_k = {\mathbf{h}}_k  {\mathbf{h}}_k^\Her$ for $K+1 \leq k \leq N_\BS$. It has been observed for ZF precoders that as the dither power is increased, both the useful signal power and the quantization noise power decrease with the decrease in quantization noise power being more significant up to some maximum dither power \cite{amodh}. After this point, the decrease in the signal power overtakes the decrease in quantization noise power thus reducing the achievable SQINR. 

In the framework considered in this paper, the amount of dither power directly depends on the power allocated to the dummy users which in turn depends on their respective SQINR constraints $[\gamma_{K+1} \dots \gamma_{N_\BS}]$. Since all the dummy users lie in the null space of the true users and there is no advantage of choosing any particular dummy user, we formulate the optimal dither power problem as a scalar optimization problem by forcing the SQINR constraint of all dummy users to be equal, i.e. $\gamma_k = \gamma_d$ for $K+1 \leq k \leq N_\BS$. With this formulation, choosing the optimal dither power boils down to choosing the correct value of the scalar $\gamma_d$. In this paper, we find the value of $\gamma_d$ by a simple grid search method. We start the procedure by initializing $\gamma_d = \epsilon_2 \approx 0$  and then increase its value up to a maximum value of $\gamma_{d_\text{max}}$ using a constant step size of $\delta_\gamma$. After solving for the optimal DL/UL power allocation vectors and the beamforming matrix using Algorithm \ref{alg:alg1}, we add the $(N_\BS - K)$ dummy users to the system with their respective channels. From here onwards in addition to computing the UL power allocation vector $\mathbf{p}$ and the beamforming matrix $\mathbf{T}$, we also compute the DL power allocation vector $\mathbf{q}$ during each iteration. Using $\mathbf{q}$, $\mathbf{T}$, and (\ref{eq:downlinkSINR_covar}), we compute the achieved DL SQINR $\gamma_k^{\text{DL}}(\mathbf{T} ,  \mathbf{Q} , \mathbf{q})$ in each iteration. The optimization procedure is stopped when the $\bar{\gamma} =  \min_{1 \leq k \leq K} \gamma_k^{\text{DL}}(\mathbf{T} ,  \mathbf{Q} , \mathbf{q})$ starts decreasing. At this point, further increase in the dither power results in more reduction in the signal power than reduction in the quantization noise power. This procedure is summarized in Algorithm \ref{alg:alg2}.

\begin{algorithm}[t]
\caption{Optimized dithering using dummy users }\label{alg:alg2}
 1) Compute $\mathbf{p}^{\star}, \mathbf{q}^{\star}, \text{and } \mathbf{T}^{\star}$ using Algorithm \ref{alg:alg1} \vspace{.1cm}
 
 2) Compute DL SQINR $ \gamma_k^{\text{DL}}(\mathbf{T}^\star ,  \mathbf{Q} , \mathbf{q}^{\star}) $ using (\ref{eq:downlinkSINR_covar}) \vspace{.1cm} 

 3) Initialize dummy users: \vspace{.1cm}  \\
 \text{} \hspace{0.5cm} $\gamma_d = \epsilon_2, \quad \gamma_{d_\text{max}}, \quad \delta_\gamma, \quad [\gamma_1 \dots \gamma_K, \underbrace{\gamma_d \dots \gamma_d}_{N_\BS - K}]$  \vspace{.1cm}  \\
 \text{} \hspace{0.5cm} ${\mathbf{H}}_d = \mathbb{N}({\mathbf{H}}), \quad \mathbf{R}_k = {\mathbf{h}}_k  {\mathbf{h}}_k^\Her \text{ for } K+1 \leq k \leq N_\BS  $ \vspace{.1cm}  \\
 \text{} \hspace{0.5cm} $\mathbf{p}^{\star(n)} = [\mathbf{p}^{\star}, \underbrace{0, \dots, 0}_{N_\BS - K}], \quad \mathbf{q}^{\star(n)} = [\mathbf{q}^{\star}, \underbrace{0, \dots, 0}_{N_\BS - K}]   $ \vspace{.1cm}  \\

 4) \textbf{ while} $(\lambda_{\max}^{(n-1)} - \lambda_{\max}^{(n)} \geq \epsilon) || (\bar{\gamma}^{(n)} \geq \bar{\gamma}^{(n-1)})$ \vspace{.05cm} 
 
 5) \hspace{0.4cm} $n = n+1$ \vspace{.05cm} 
  
 6) \hspace{0.4cm}   $\forall_k \text{ } {\mathbf{t}}_k^{\star(n)} = \mathbf{v}_{\max}\left( \mathbf{R}_k,  \mathbf{Q}_k(\mathbf{p}^{(n-1)})\right)$  \vspace{0.05cm} 
 
 7) \hspace{0.4cm}  $\forall_k \text{ } {\mathbf{t}}_k^{\star(n)} = \mathbf{t}_k^{\star(n)} / \| \mathbf{t}_k^{\star(n)} \|_2 $ \vspace{0.05cm}  
 
 8) \hspace{0.4cm}  $\mathbf{\Lambda}({\mathbf{T}}^{\star(n)}, P_\BS) \mathbf{p}_\text{ext}^{\star(n)} =  \lambda_{\max}^{(n)} \mathbf{p}_\text{ext}^{\star(n)}$  \vspace{0.05cm}
 
 9) \hspace{0.4cm}  $\mathbf{p}^{\star(n)} = \mathbf{p}_\text{ext}^{\star(n)}[1,\dots,N_\BS] / \mathbf{p}_\text{ext}^{\star(n)}[N_\BS+1]$ \vspace{.05cm}
 
10) \hspace{0.2cm} $\mathbf{\Upsilon}({\mathbf{T}}^{\star(n)}, P_\BS) \mathbf{q}_\text{ext}^{\star(n)} =  \lambda_{\max}^{(n)} \mathbf{q}_\text{ext}^{\star(n)} $ \vspace{.05cm}

11) \hspace{0.2cm} $\mathbf{q}^{\star(n)} = \mathbf{q}_\text{ext}^{\star(n)}[1,\dots,N_\BS] / \mathbf{q}_\text{ext}^{\star(n)}[N_\BS+1]$ \vspace{.05cm}

12) \hspace{0.2cm} Compute DL SQINR $\gamma_k^{\text{DL}(n)}(\mathbf{T}^{\star(n)} ,  \mathbf{Q} , \mathbf{q}^{\star(n)}) $ by (\ref{eq:downlinkSINR_covar}) \vspace{.05cm} 

13) \hspace{0.2cm} $\bar{\gamma}^{(n)} = \min_{1 \leq k \leq K} \gamma_k^{\text{DL}(n)}(\mathbf{T}^{\star(n)} ,  \mathbf{Q} , \mathbf{q}^{\star(n)})$ \vspace{.05cm}

14) \hspace{0.2cm} \textbf{if} $ (\bar{\gamma}^{(n)} \geq \bar{\gamma}^{(n-1)})$ \vspace{.05cm}

15) \hspace{0.2cm} \quad $\gamma_d = \gamma_d + \delta_\gamma$ \vspace{.05cm}
 
16) \hspace{0.2cm} \textbf{end} \vspace{.0cm}

17) \textbf{end} \vspace{.0cm}
\end{algorithm}

We now turn our attention towards tuning the hyperparameters $\gamma_{d_\text{max}}$ and $\delta_\gamma$ in Algorithm \ref{alg:alg2}. To this end, we look at a toy example in Remark \ref{rmk:rmk0}.

\begin{remark}\label{rmk:rmk0}
Consider the following DL covariance matrix with large off-diagonal entries.
\begin{equation*}
\mathbf{R}_{\mathbf{y}_\text{d}} = \begin{bmatrix}
    1 \quad 0.99 \\ 
     0.99 \quad 1
      \end{bmatrix}.
\end{equation*}
Using (\ref{eq:upn2}) and (\ref{eq:upn3}), it can be verified that the DL quantization noise matrix is given by
\begin{equation*}
\mathbf{R}_{\pmb{\eta}_\text{d}} \approx \left(1 - \frac{2}{\pi} \right) \begin{bmatrix}
    1 \quad 0.99 \\ 
     0.99 \quad 1
      \end{bmatrix}.
\end{equation*}
Now, assume that we add dithering to the system such that it gives equal contribution to the diagonal elements as the original all ones diagonal of $\mathbf{R}_{\mathbf{y}_\text{d}}$ and minimal contribution to the off-diagonal elements 
\begin{equation*}
\mathbf{R}_{\mathbf{y}_\text{d}} \approx \begin{bmatrix}
    2 \quad 1 \\ 
     1 \quad 2
      \end{bmatrix}.
\end{equation*}
Using (\ref{eq:upn2}) and (\ref{eq:upn3}), $\mathbf{R}_{\pmb{\eta}_\text{d}}$ is now given by
\begin{equation*}
\mathbf{R}_{\pmb{\eta}_\text{d}} \approx \left(1 - \frac{2}{\pi} \right) \begin{bmatrix}
    1 \quad 0.04 \\ 
    0.04 \quad 1
      \end{bmatrix}.
\end{equation*}
It can be seen from this toy example that the off-diagonal elements of the DL correlation matrix $\mathbf{R}_{\mathbf{y}_\text{d}}$ need to be less than or equal to half of the diagonal elements for the quantization noise to be approximately uncorrelated. Another way to look at this is to observe that $\text{sin}^{-1}(x) \approx x$ for $|x| \leq 0.5$. In case of a single active user, $\mathbf{R}_{\mathbf{y}_\text{d}} = {q_1} \mathbf{t}_1^\ast \mathbf{t}_1^\T$ and $q_1 = P_\BS$. Based on this intuition, approximately half of the total power $P_\BS$ should be allocated to the dummy users for the resulting quantization noise to be uncorrelated. The power allocated to true/dummy users depends on the SQINR constraints $[\gamma_1, \gamma_d \dots \gamma_d] \in \mathbb{R}^{\NBS+}$, the channel matrix and the beamformer matrix through (\ref{eq:psiMatrix2}). Given the symmetry of the whole power allocation and beamformer optimization procedure, it is reasonable that the SQINR constraint should be equal to that of the true user divided equally among all the $\NBS - 1$ dummy users. 
\end{remark}

We numerically found out that for $K = 1$,  the maximum value of the SQINR constraint $\gamma_{d_\text{max}}$ for the dummy users is on the order of the true user SQINR $\gamma_1$ divided equally into the $N_\BS - 1$ dummy users. This agrees with the intuition developed in remark \ref{rmk:rmk0}. For $K > 1$, the amount of dithering needed for the small angle approximation to hold true is less than the single user case as is apparent from Fig. \ref{fig:diag}. Since $\gamma_{d_\text{max}}$ is an upper bound, we choose $\gamma_{d_\text{max}} = \max_{1 \leq k \leq K} \frac{\gamma_k}{N_\BS - K}$ for our simulation results in Section \ref{sec:results}. The step size $\delta_\gamma$ allows a tradeoff between computational complexity and the accuracy of the optimal dither power. A larger value of $\delta_\gamma$ will result in overshooting/undershooting the optimal value by a bigger margin and reduced computational complexity. With $N_\text{max}$ denoting the maximum number of iterations of the optimization procedure, we choose $\delta_\gamma = \frac{\gamma_{d_\text{max}}}{N_\text{max}}$. This makes the maximum computation cost of the dithering procedure fixed. For our simulation results, we choose $N_\text{max} = 16$. The effect of adding dummy users to the UL-DL framework on the diagonal metric $d( \mathbf{R}_{\pmb{\eta}_\text{d}} )$ is shown in red (\tikz\draw [red,thick] (0,0) -- (.5,0);) in Fig. \ref{fig:diag}. It can be observed that after adding the dummy users $\mathbf{R}_{\pmb{\eta}_\text{d}}$ is always diagonal irrespective of the number of users and hence the quantization noise is uncorrelated.

\begin{remark}\label{rmk:rmk1}
The choice of $\gamma_{d_\text{max}}$, $\delta_\gamma$ and the optimization procedure for the scalar $\gamma_d$ is based on heuristics. Other procedures based on binary search, backtracking and varying the constant step size $\gamma_d$ can be designed without affecting the results presented in Section \ref{sec:results}.
\end{remark}

%To this end, we focus on the extreme case of a single active user in Claim \ref{claim:claim1}.
%\begin{claim}\label{claim:claim1}
%For $K = 1$, the maximum value of the SQINR constraint for the dummy users is on the order of the true user SQINR $\gamma_1$.
%\end{claim}
%\begin{proof}
%See Appendix \ref{sec:C}.
%\end{proof}

%%%%%%%%%%%%%%%%%%%%%%%%%%%%%%%%%%%%%%%%%%%%%%%%%%%%%%%%
\section{Results and discussion} \label{sec:results}
In this section, we first describe our simulation setup followed by a brief description of the benchmark strategies used to compare with the proposed algorithm. We then present the SQINR and uncoded BER results.

%%%%%%%%%%%%%%%%%%%
\vspace{-0.3cm}
\subsection{Simulation setup} \label{sec:setup}
Our simulations are done based on the 3GPP Urban-Macro (UMa) line-of- sight (LoS) channel model (3GPP 38.901 UMa LoS) generated through Quadriga \cite{Quadriga}. For the results presented in this paper, we consider a scenario where the individual users are distributed over a $120^\circ$ sector around the BS  from a minimum range of 50 m to a maximum range of 150 m according to a uniform random variable. For the proposed algorithm, the target SQINRs are set equal to 3 dB for all users. The channel parameters generated by Quadriga for each realization are converted to a complex baseband channel using a truncated sinc pulse shape. All antenna elements at the BS and the users have an omni-directional pattern with a gain of 0 dBi. The important simulation parameters (unless otherwise specified) are summarized in Table \ref{table:table2}. %We point out here that the total transmit power $P_\BS$ and noise variance corresponding to the bandwidth of 1 GHz are for all $N_\text{SC}$ subcarriers. Since we are limited to flat fading channels in this work, the transmit power and the noise variance are equally divided among the $N_\text{SC}$ subcarriers and are a factor of 15 dB (= 32) less than what is shown in Table \ref{table:table2}. We then use an OFDM formulation with $N_\text{SC} = 32$ subcarriers and a cyclic prefix length of $N_\text{CP} = 8$ to convert the multi-tap frequency selective channel into multiple frequency flat channels and use one of the resulting OFDM subcarriers for our simulations.

%The simulations are done at 60 GHz with a $P_\text{max}$ of 40 dBm and $N_\text{BS} = 32$.

\begin {table}[h]
\begin{center}
\begin{tabular}{ | >{\centering\arraybackslash} m{3.6cm} | >{\centering\arraybackslash} m{3.2cm} | }
  \hline 			
  {Quadriga channel model} & {$\text{3GPP 38.901 UMa LoS}$} \\  			
  \hline 
  { Number of antennas $N_\BS$} & 128 \\  			
  \hline 
  { $\text{Antenna element pattern}$} & $\text{omni-directional}$ \\  			
  \hline 
  { Total transmit power $P_\BS$} & 24 dBm \\  			
  \hline 
  {Carrier frequency $f_c$} & 60 GHz \\  			
  \hline 
  {Bandwidth $B$} & 8 MHz \\  			
  \hline 
%  {Number of subcarriers $N_\text{SC}$} & 32 \\  			
%  \hline 
%  {Cyclic prefix length $N_\text{CP}$} & 8 \\  			
%  \hline 
\end{tabular}
\end{center}
\vspace{-0.1cm}
\caption{Important simulation parameters.}
\vspace{-0.5cm}
\label{table:table2} 
\end{table}

%%%%%%%%%%%%%%%%%%%
\vspace{-0.0cm}
\subsection{Benchmark strategies}
We use ZF precoding \cite{amodh} as a benchmark for our proposed technique with the DL BF matrix given by $\mathbf{T} =  {\mathbf{H}}^\Her ({\mathbf{H}} {\mathbf{{H}}}^\Her)^{-1}$. We choose three ways to allocate the per-antenna power allocation matrix $\mathbf{Q}$ in the 1-bit system given~by 

\begin{itemize}

\item \emph{ZF Opt-Pwr}: $\mathbf{Q} = \diag \left(\sum_{k=1}^K q_k \widehat{\mathbf{t}}_k \widehat{\mathbf{t}}_k^\Her \right)^{\frac{1}{2}}$ where the power allocation vector $\mathbf{q}$ is obtained using the optimal DL power allocation procedure described in \ref{sec:DLpower} and ${\mathbf{T}}^\star = [ \widehat{ \mathbf{t} }_1 \dots \widehat{ \mathbf{t} }_K ]$ for $\widehat{ \mathbf{t} }_k= { \mathbf{t} }_k / \| { \mathbf{t} }_k \|_2$. We note here that ZF with optimal power allocation has not been considered in literature before to the best of authors' knowledge.

\item \emph{ZF ZF-Pwr}: Let $\widehat{\mathbf{q}} = \left[ \| \mathbf{t}_1 \|_2^2 \dots \| \mathbf{t}_K \|_2^2 \right]$ and ${\mathbf{q}} = \frac{P_\BS}{\sum_{k=1}^K \widehat{\mathbf{q}}_k }\widehat{\mathbf{q}}$ be its scaled version normalized to $P_\BS$. The choice of $\mathbf{Q} = \diag \left(\sum_{k=1}^K q_k \widehat{\mathbf{t}}_k \widehat{\mathbf{t}}_k^\Her \right)^{\frac{1}{2}}$ then makes the per-antenna power allocation after the 1-bit quantization operation to be the same as the per-antenna power in the $\infty$-resolution case.

\item \emph{ZF Equal-Pwr}: $\mathbf{Q} = \diag \left( {\frac{P_\BS}{\NBS}} \mathbf{1}_{\NBS} \right)^{\frac{1}{2}}$. This corresponds to sending the same amount of power from all $\NBS$ antennas. Such power allocation has been considered before in existing literature \cite{amodh,Studer}.

\end{itemize}
For a fair comparison, we also add optimized dithering to the system when using ZF precoding using a procedure similar to described in \cite{amodh}. IID Gaussian noise with variance $\sigma_d^2$ is projected onto the null space of the channel matrix ${\mathbf{H}}$ given by $ \left( \mathbf{I}_{} - {\mathbf{H}}^\Her ({\mathbf{H}} {\mathbf{{H}}}^\Her)^{-1} {\mathbf{H}} \right)$ and added to the linearly precoded transmit signal before the 1-bit quantization operation. Similar to the procedure described in Section \ref{sec:dummy}, we start from a small $\sigma_d^2$ and keep increasing it using a constant step size till the minimum SQINR starts decreasing.
 
We also compare with a non-linear precoding solution called SQUID. The algorithm SQUID is based on Douglas-Rachford splitting of a squared $\ell_\infty$-norm relaxation of the symbol MMSE problem. We refer the reader to \cite{Studer} for a more detailed description of the algorithm. The hyperparameters involved in the implementation of SQUID were chosen according to the guidelines given in \cite{Studer}. In addition to the complexity arising from the size of the problem being solved, one major drawback of SQUID (and other non-linear algorithms) is that the optimization problem has to be solved for \emph{every} transmit symbol vector sent out during the coherence time of the channel. In the results that follow, the proposed Algorithm \ref{alg:alg1} is denoted as `Opt' whereas its generalization using dummy users (Algorithm~\ref{alg:alg2}) is denoted as `Opt Dummy'. Algorithm~\ref{alg:alg2} with equal per-antenna power allocation is denoted by `Opt Dummy Equal-Pwr'.
 
%ADD STATISTICAL INFO 

%%%%%%%%%%%%%%%%%%%
\vspace{-0.2cm}
\subsection{SQINR results}\label{sec:SQINR}

We use the ergodic sum rate given by $\mathbb{E}\left[ \sum_{k=1}^K \text{log}_2 ( 1 + \gamma_k^{\text{DL}} ) \right]$ and the ergodic minimum rate given by $\mathbb{E}[ \min_{1\leq k \leq K} \text{log}_2 ( 1 + \gamma_k^{\text{DL}} ) ]$ as the metrics of choice for our results. $\gamma_k^{\text{DL}}$ is calculated using (\ref{eq:downlinkSINR_covar}) with the exact \emph{arcsine} law and without the small angle approximation. The expectation is computed by averaging over multiple IID realizations of the user positions and channel realizations.
%Also, all computations in the proposed Algorithms \ref{alg:alg1} and \ref{alg:alg2} (and other benchmark algorithms) are done using the \emph{estimated} channel matrix $\widehat{\mathbf{H}}$.

The ergodic sum rate (in blue on the left y-axis) and ergodic minimum rate (in red on the right y-axis) are shown in Fig.~\ref{fig:meanCap} as a function of the number of users operating simultaneously. The proposed strategy and ZF perform fairly close to each other in terms of the sum rate when the number of users is small. For $K > 5$, it can be seen that the two strategies start to diverge and the performance of ZF starts to degrade. The proposed algorithm takes into account the MUI and the quantization noise for designing the beamformers and outperforms ZF by more than 10 b/s/Hz when $K = 15$. This behavior is fairly representative of ZF precoding and has also been observed with $\infty-$resolution ADCs/DACs. It can also be observed that the per-antenna power allocation $\mathbf{Q}$ has a big impact on the ZF precoding performance with the optimal power allocation outperforming the ZF power allocation by 2-3 b/s/Hz in terms of sum rate. For $K$ small, it can be seen that adding dummy users makes a significant difference to the sum and minimum rate. The proposed algorithm also outperforms ZF precoding in terms of the minimum rate by 0.75 b/s/Hz for larger number of users. Though seemingly small, this improvement can be very important from an outage probability and fairness perspective. From here onwards, we are only going to focus on the case with the dummy users present in the system.

\begin{figure}[t]
    	\begin{center}
    		\includegraphics[width=.47\textwidth,clip,keepaspectratio]{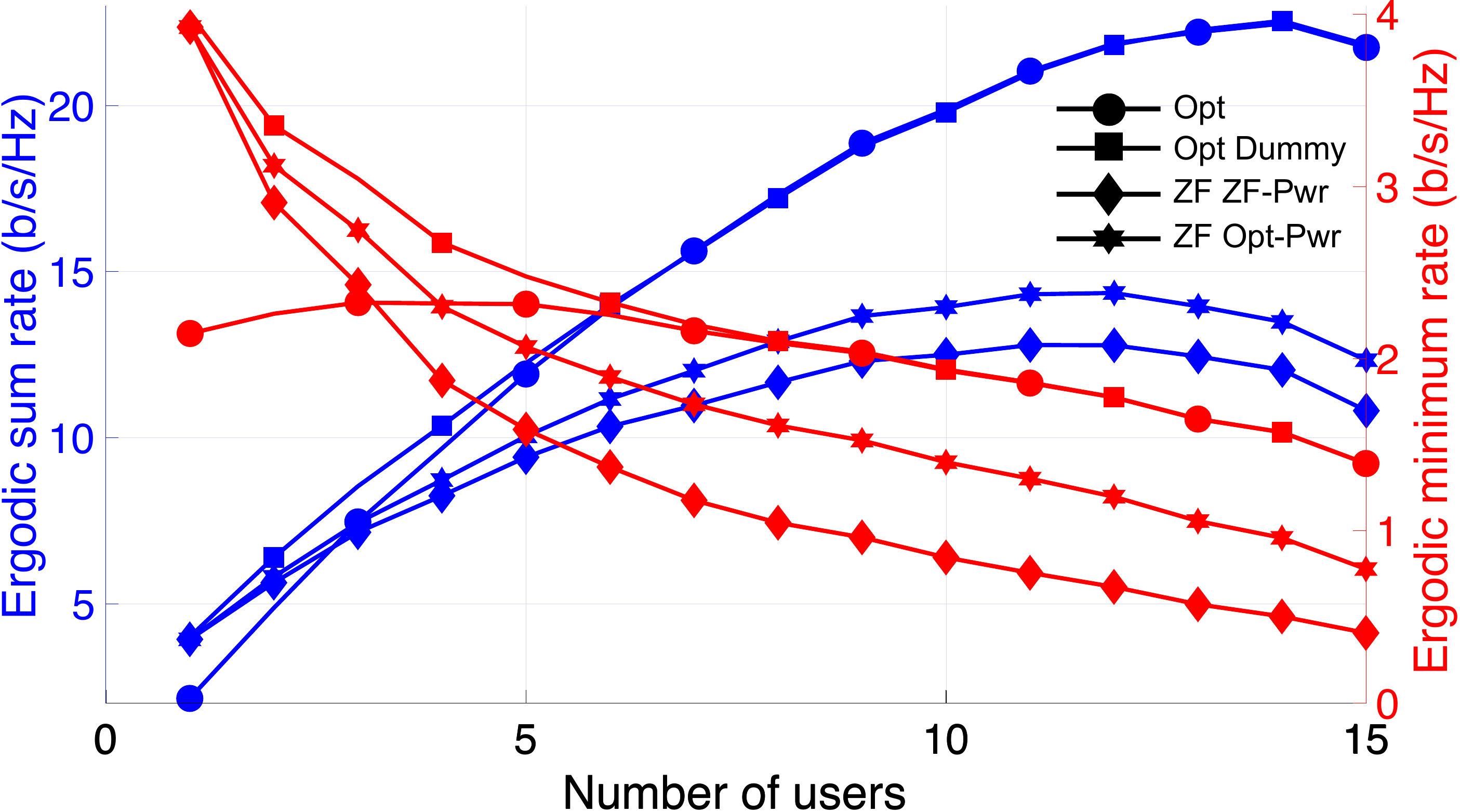}
    	\end{center}
    	\caption{Ergodic sum rate (in blue on the left y-axis) and ergodic minimum rate (in red on the right y-axis) versus number of users. As the number of users increases, the proposed algorithm outperforms ZF by more than 10 b/s/Hz and 0.75 b/s/Hz in terms of the ergodic sum and ergodic minimum rate. }%Best viewed in color.}
\vspace{-0.4cm}
    	 \label{fig:meanCap}
\end{figure}

\begin{figure}[b]
    	\begin{center}
    		\includegraphics[width=.47\textwidth,clip,keepaspectratio]{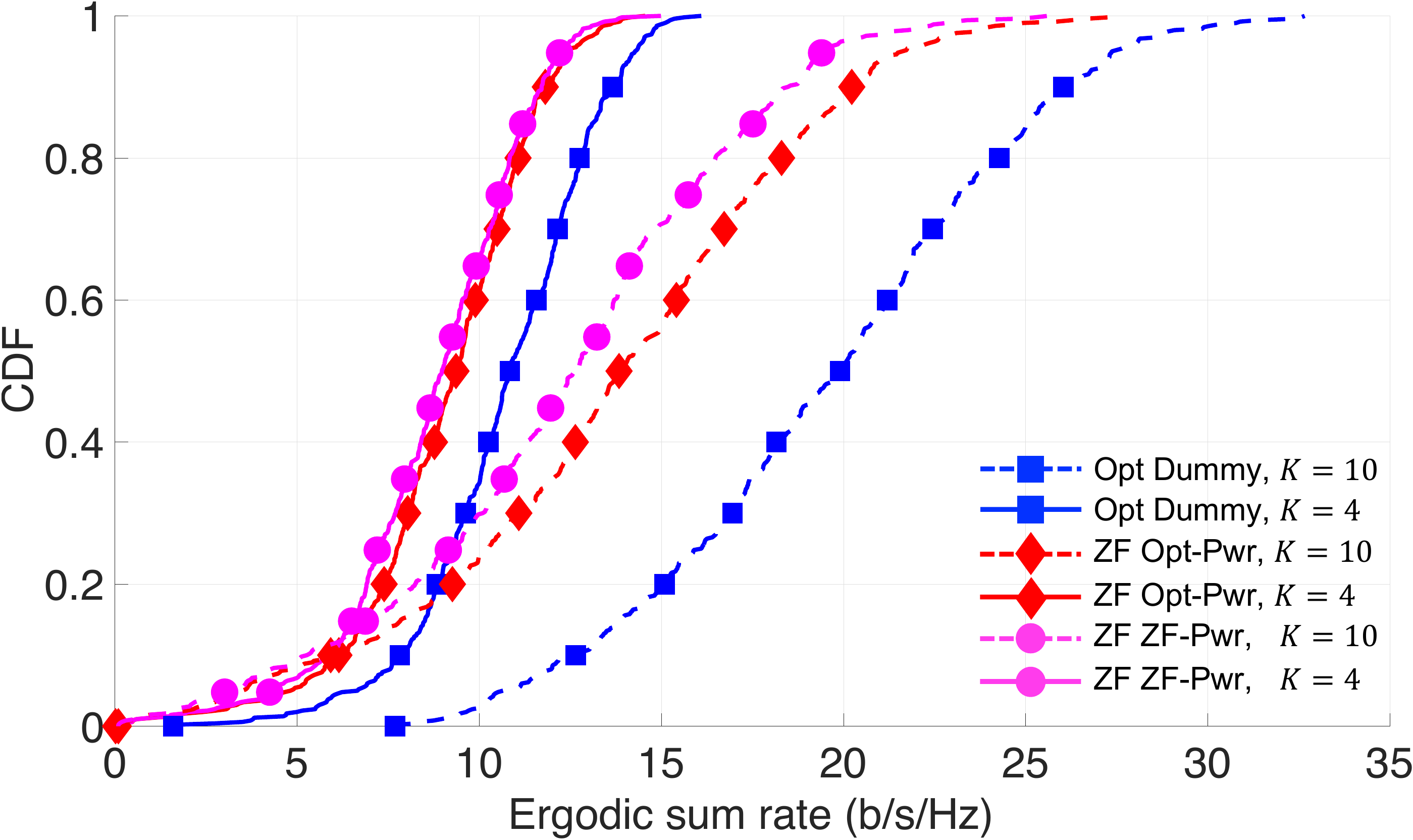}
    	\end{center}
    	\caption{CDF of the ergodic sum rate for $K = 4$ (in solid lines) and $K = 10$ (in dashed lines). The proposed algorithm performs better than ZF in both settings. The difference, however, is more pronounced when the number of users increases due to the careful design of the  precoders which takes into account the MUI and the quantization noise.}% Best viewed in color.}
    	 \label{fig:cdfCap}
\vspace{-0.2cm}
\end{figure}
Fig. \ref{fig:cdfCap} illustrates the cumulative distribution function (CDF) of the ergodic sum rate for the proposed algorithm and ZF for $K = 4$ in solid lines and for $K = 10$ in dashed lines. It can  be observed that the proposed algorithm performs better than ZF in both cases. For $K = 4$, the difference between the two algorithms is marginal. For $K = 10$, our solution does better than ZF by about 5 b/s/Hz over the whole distribution space.

Next, we look at the performance as the number of BS antennas $\NBS$ change. The ergodic sum rate is plotted in Fig. \ref{fig:diffAnt} for $\NBS \in [32, 64,128]$. It can be observed that the proposed algorithm outperforms ZF in all three cases. The relative difference in sum rate, however, becomes smaller as the number of BS antennas becomes larger. This behavior is due to improved spatial resolution of the BS due to the larger aperture and has been observed before in context of~massive MIMO where low complexity precoders (such as ZF~and~MRT) were shown to achieve almost optimal performance \cite{ULOFDMStuder,Studer}.

\begin{figure}[t]
    	\begin{center}
    		\includegraphics[width=.47\textwidth,clip,keepaspectratio]{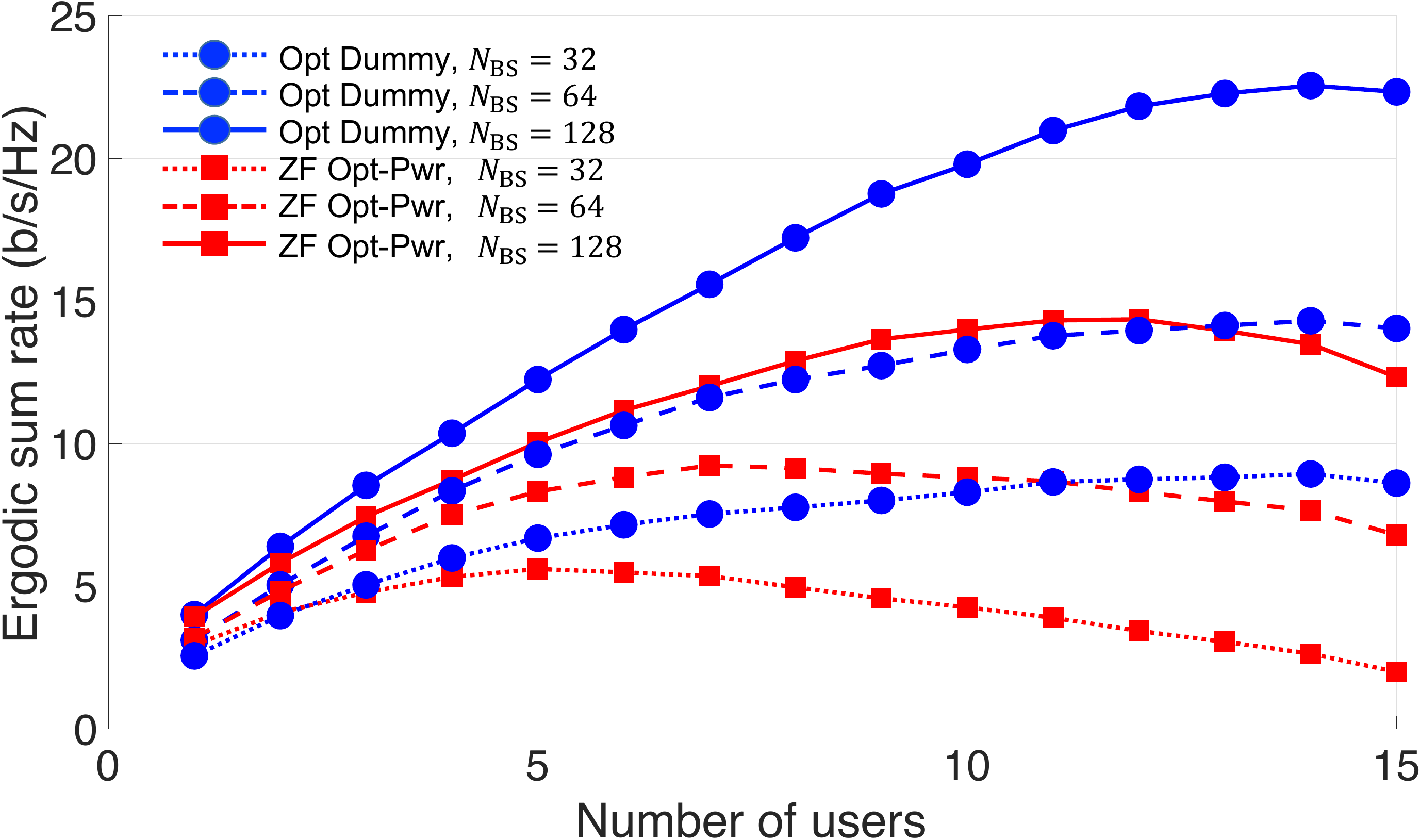}
    	\end{center}
\vspace{-0.2cm}
    	\caption{Ergodic sum rate versus number of users for $\NBS \in [ 32, 64,128]$. The proposed algorithm outperforms ZF in all three cases. The relative gap, however, becomes smaller as the number of BS antennas increases.} %Best viewed in color.}
    	 \label{fig:diffAnt}
\vspace{-0.5cm}
\end{figure}

%The uncorrelated quantization noise (used in proving UL DL duality) was a crucial assumption in this paper. This assumption is not always true and can lead to significantly different results. Our ongoing work will generalize the presented framework to settings where this assumption is not true by adding optimized dithering to the system. We will also explore the setting where only the statistical CSI (in contrast to the full CSI in this paper) is available at the BS.

%\begin{figure}[h]
%    	\begin{center}
%    		\includegraphics[width=.5\textwidth,clip,keepaspectratio]{figsK/diffPwr.pdf}
%    	\end{center}
%    	\caption{Ergodic sum rate versus number of users for $P_\BS \in [30,35,40]$ dBm. It can be seen that the performance gap between the proposed solution and ZF becomes larger as the total transmit power increases due to correlated interference and quantization noise at higher powers. Best viewed in color.}
%    	 \label{fig:diffPwr}
%\end{figure}

\begin{figure}[b]
    	\begin{center}
    		\includegraphics[width=.47\textwidth,clip,keepaspectratio]{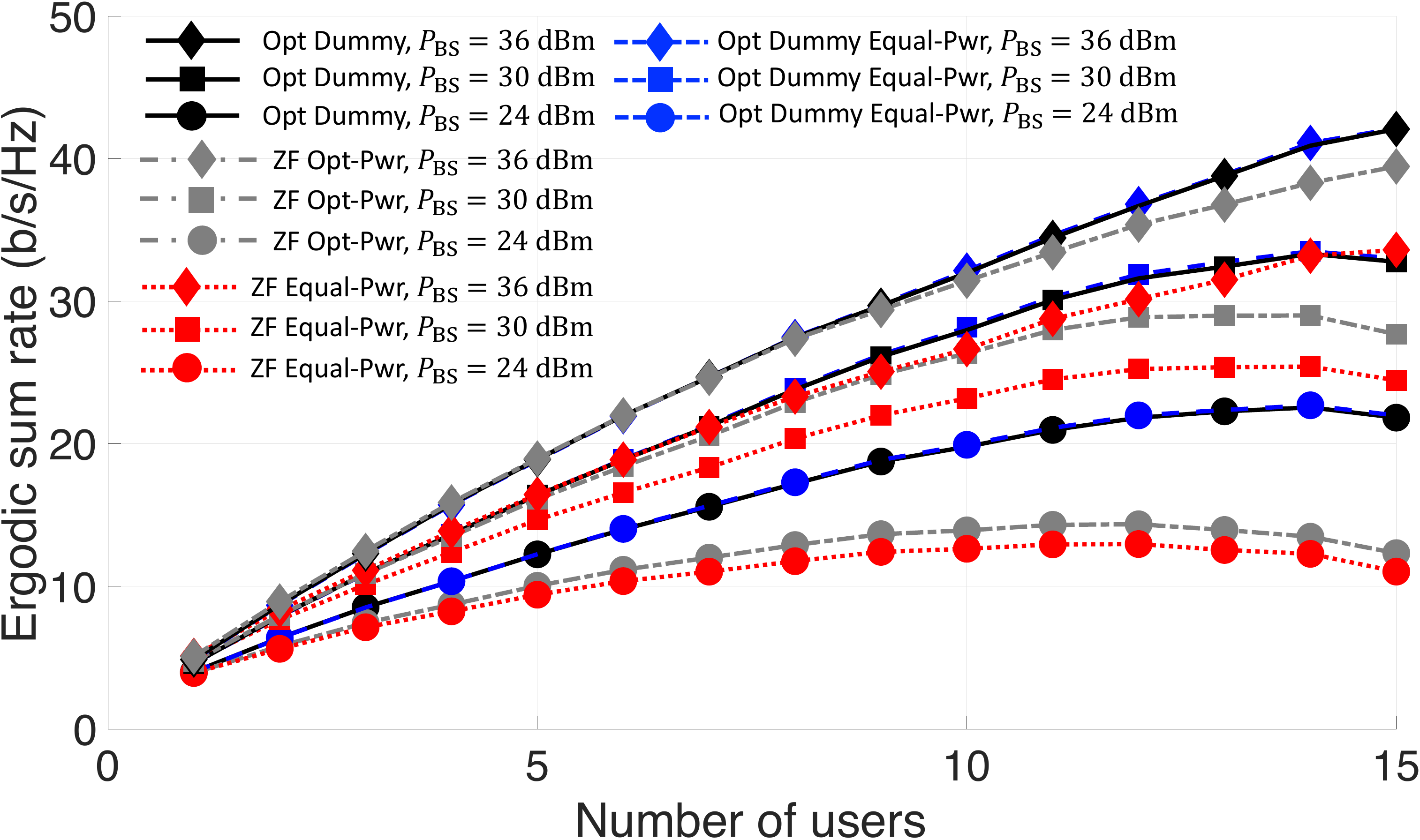}
    	\end{center}
\vspace{-0.2cm}
    	\caption{Ergodic sum rate for the proposed solution and ZF with \emph{optimal} and \emph{equal} per-antenna power allocation versus number of users for $P_\BS \in [24,30,36]$ dBm. Equal per-antenna power allocation significantly degrades the performance for ZF.}
    	 \label{fig:equalPerAnt}
\end{figure}

Fig. \ref{fig:equalPerAnt} illustrates the ergodic sum rate (for both the proposed solution and ZF) for $P_\BS ~\in~[24,30,36]$ dBm for two choices of the power allocation matrix $\mathbf{Q}$: optimal power allocation and equal per-antenna power allocation. The key takeaway~for the optimal power allocation setting is that the proposed algorithm (in black) outperforms ZF (in grey) over the entire SQINR region. The difference, however, seems to decrease with the increase in transmit power. This is not the case with a smaller array aperture where the proposed solution actually increasingly outperforms ZF with the increase in power. 
%This is due to the MUI and quantization noise becoming correlated at higher transmit powers. In contrast to ZF which minimizes the MUI, the proposed solution takes into account both the MUI and the quantization noise resulting in better performance across the whole SQINR region. It can also be observed from Fig. \ref{fig:equalPerAnt} that the performance of ZF precoding actually deteriorates at higher powers for~large number of users compared to when the transmit power is smaller.

The optimal power allocation procedure from Section \ref{sec:DLpower} requires amplifiers that are linear over the input signal dynamic range. The linearity of the power amplifier over the dynamic range of the input signal is a critical issue especially at mmWave frequencies. An equal per-antenna power allocation is a useful solution to reduce hardware complexity and increase efficiency by designing power amplifiers that operate in their saturation region at a fixed power point. The ergodic sum rate for ZF (in red) and the proposed solution (in blue) for the equal per-antenna power allocation setting is illustrated in Fig. \ref{fig:equalPerAnt}. It can be seen that ZF precoding suffers a degradation in performance compared to the optimal power allocation setting show in grey. Furthermore, this degradation in performance increases with the increase in power. The proposed solution under equal per-antenna power allocation achieves the same performance as the optimal power allocation. This observation thus results in further reduction in hardware complexity from a power amplifier design perspective.

%%%%%%%%%%%%%%%%%%%
\vspace{-0.2cm}
\subsection{BER results} \label{sec:BER}
We now present results in terms of the uncoded BER for transmit symbols drawn from unit-norm normalized QPSK and 16-Quadrature Amplitude Modulation (16-QAM) constellations. The BER results are computed by transmitting 100 IID (QPSK or 16-QAM) symbols for each user in the system for each channel realization. The received symbols are then decoded using a minimum distance decoder \cite{hela} and then mapped to bits. The BER is computed by averaging over the 100 symbols, the channel realizations and the number of active users. For QPSK, the minimum distance decoder is agnostic to the received symbol amplitude and can be implemented by just choosing the quadrant in which the symbol lies \cite{Studer}. For 16-QAM, the received symbols need to be scaled appropriately. We implement a blind estimation method based on the work in \cite{hela} where each users estimates a scaling factor $g_k$ using a block of the received symbols before the decision operation.

\begin{figure}[t]
    	\begin{center}
    		\includegraphics[width=.46\textwidth,clip,keepaspectratio]{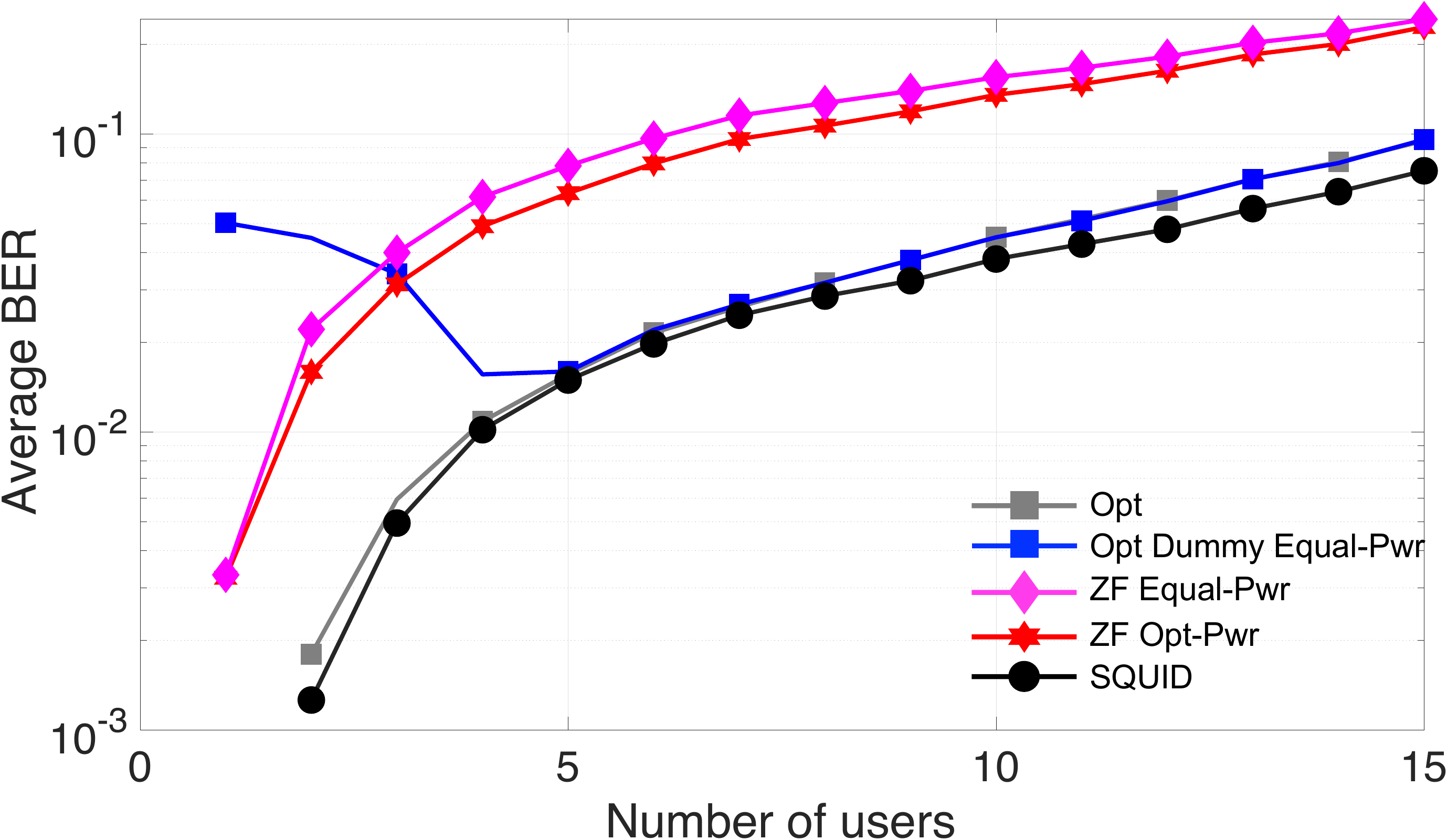}
    	\end{center}
	\vspace{-0.1cm}
    	\caption{Uncoded BER performance for QPSK constellation. The proposed solution and SQUID achieve similar performance outperforming ZF by a considerable margin.}% Best viewed in~color.
	\vspace{-0.3cm}
    	 \label{fig:ber4}
\end{figure}

The uncoded BER for the proposed solutions and benchmark strategies is plotted in Fig. \ref{fig:ber4} against the number of active users for symbols drawn from the QPSK constellation. The first observation from Fig. \ref{fig:ber4} is that adding dummy users to the proposed algorithm results in slightly degraded performance for small number of users. This is in contrast to the SQINR results presented in Section \ref{sec:SQINR} and can be slightly misleading. Adding dummy users does improve the performance. This performance improvement, however, results from the improved dynamic range which does not matter for the QPSK constellation. The dummy users slightly perturb the phase of the signal thus resulting in a small loss in performance. It can also be seen that the proposed solution and SQUID achieve similar performance significantly outperforming ZF. 

\begin{figure}[t]
    	\begin{center}
    		\includegraphics[width=.47\textwidth,clip,keepaspectratio]{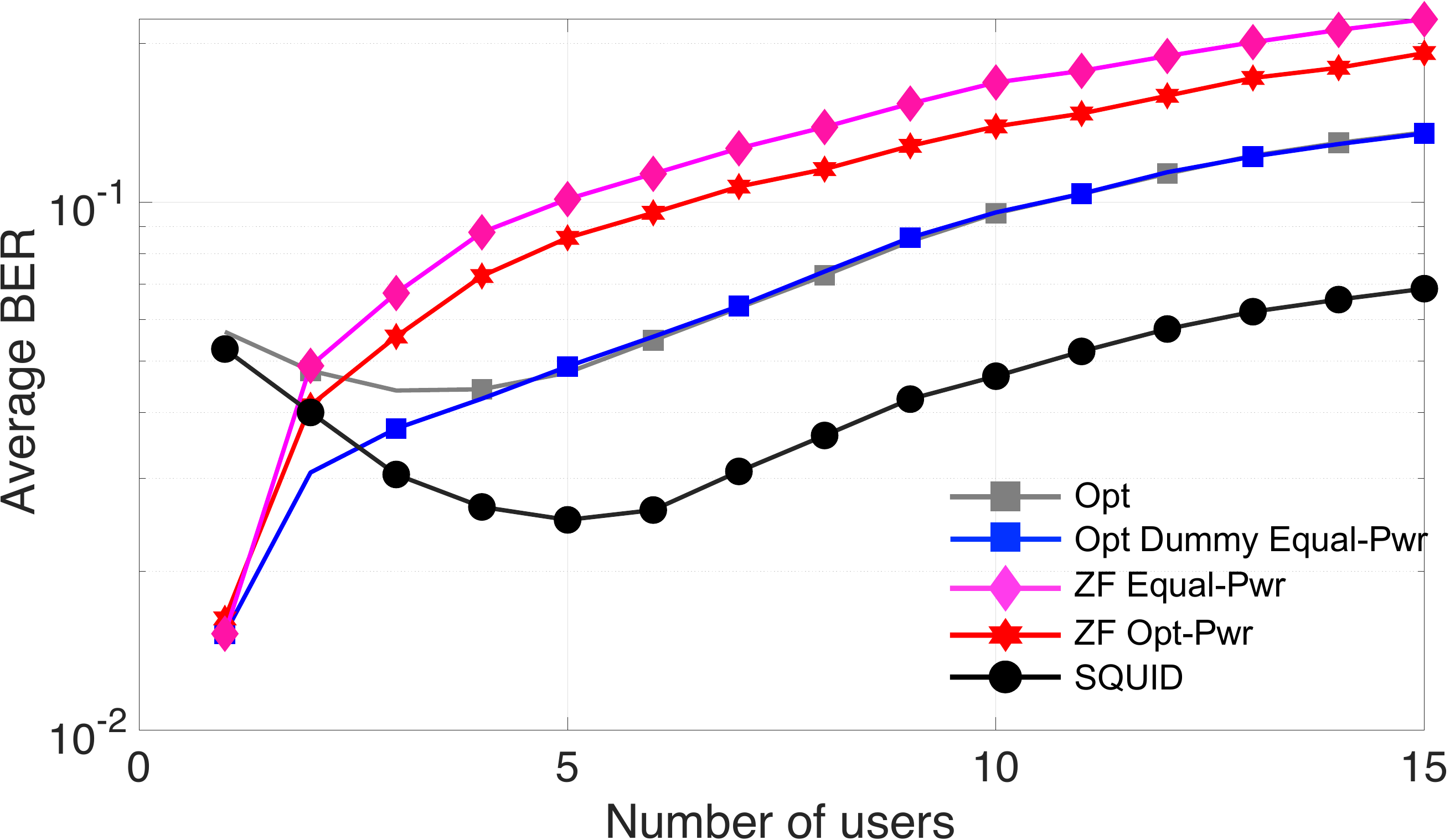}
    	\end{center}
	\vspace{-0.1cm}
    	\caption{Uncoded BER performance for 16-QAM constellation for $P_\BS = 30$ dBm. SQUID performs better than the~proposed solution except for small number of users. Adding dummy users to the proposed solution results in improved performance as observed~in SQINR results for small number of users.}% Best viewed in~color.
	\vspace{-0.4cm}
    	 \label{fig:ber16}
\end{figure}

Fig. \ref{fig:ber16} illustrates the uncoded BER (for $P_\BS = 30$ dBm) for transmit symbols drawn from the 16-QAM constellation.  It can be seen that the non-linear precoder SQUID now performs slightly better than the proposed solutions except for when the number of users is small. The performance difference, however, is reduced compared to ZF based precoding. It can also be observed that the proposed algorithm now performs better with the dummy users in the system for a smaller number of active users. This is due to the improved dynamic range as a result of introducing optimized dithering. ZF with equal per-antenna power allocation achieves the worst performance among all algorithms. 

Next, we look at the performance in terms of uncoded BER for transmit symbols drawn from the QPSK constellation. The BER as a function of the transmit power is illustrated in Fig. \ref{fig:berPwr}  for $K = 10$. It can be seen that the proposed solution and SQUID achieve similar performance significantly outperforming ZF precoding. Another observation is that all of the precoding strategies saturate at a certain level with further increase in transmit power yielding no improvement. The performance of SQUID actually starts to degrade with further increase in power and requires hyperparameter tuning for meaningful performance.

\begin{figure}[t]
    	\begin{center}
    		\includegraphics[width=.47\textwidth,clip,keepaspectratio]{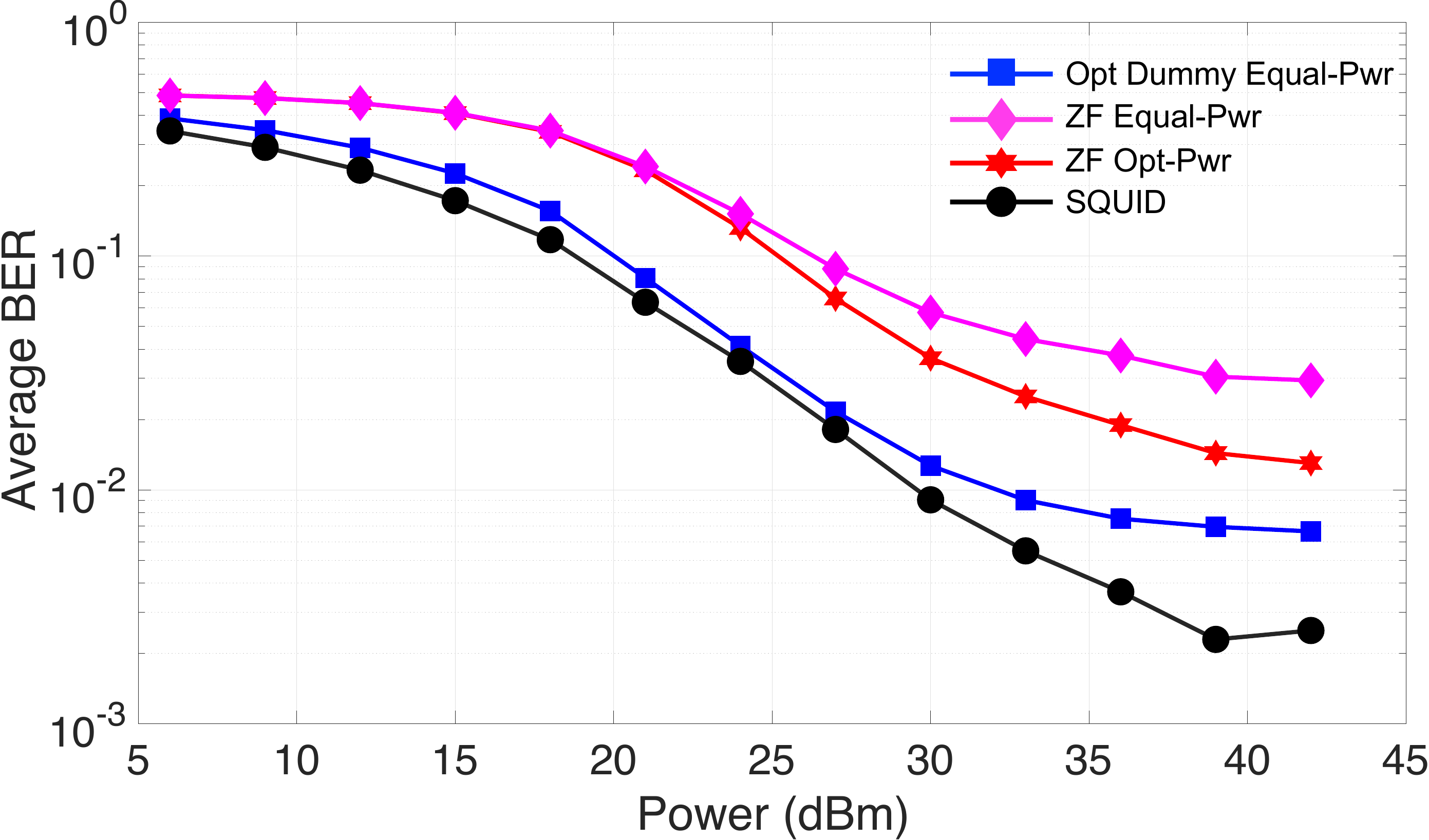}
    	\end{center}
	\vspace{-0.2cm}
    	\caption{Uncoded BER performance for QPSK constellation with $K = 10$ against $P_\BS$. The proposed solution and SQUID achieve similar performance significantly outperforming ZF.}
	\vspace{-0.4cm}
    	 \label{fig:berPwr}
\end{figure}

We have so far assumed the availability of perfect CSI at the BS. This is, however, not a realistic assumption especially in the setting where the BS is equipped with 1-bit converters. Channel estimation with 1-bit converters is a closely related problem with a rich existing literature \cite{SE,ULOFDMStuder,channelEst}. We demonstrate robustness of the proposed solution to channel estimation errors by considering a normalized IID Gaussian noise perturbed channel of the form
\begin{equation*}
\widehat{\mathbf{h}}_k = \sqrt{1-\alpha}\mathbf{h}_k + \sqrt{\alpha}\mathbf{z},
\end{equation*}
where $\mathbf{z} \sim \mathcal{CN}(0,\frac{\| \mathbf{h}_k \|^2_2}{\NBS})$ and $\alpha$ varies from 0 to 1. Similar evaluations have been considered before in \cite{hela,Studer}. The uncoded BER for symbols drawn from the QPSK constellation for $K = 10$ and $P_\BS = 30$ dBm is illustrated in Fig. \ref{fig:CEE} as a function of the normalized channel estimation error. It can be observed that the performance of the proposed solution degrades in proportion to the channel estimation error similar to the other strategies. The proposed solution performs quite close to SQUID significantly outperforming ZF  precoding.

\begin{figure}[b]
    	\begin{center}
    		\includegraphics[width=.47\textwidth,clip,keepaspectratio]{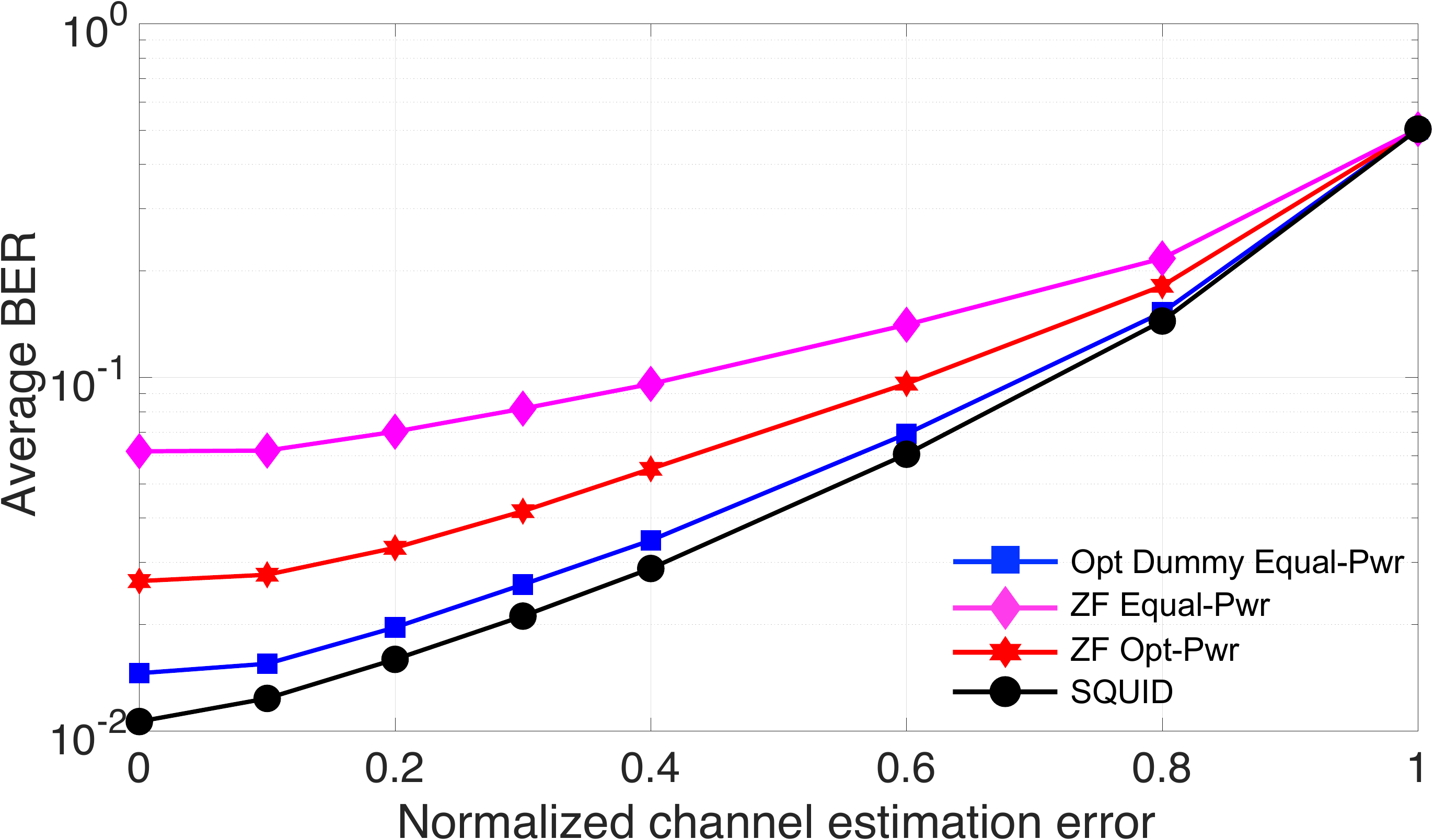}
    	\end{center}
    	\caption{Uncoded BER performance for QPSK for $K = 10$ and $P_\BS = 30$ dBm. The proposed solution suffers a degradation in performance proportional to other strategies with increasing channel estimation error. Opt Dummy and SQUID achieve almost similar performance.}
    	 \label{fig:CEE}
\end{figure}

Our numerical results in Section \ref{sec:SQINR} and \ref{sec:BER} demonstrate that the proposed solution achieves superior performance in comparison to existing linear precoding solutions in the literature. It also achieves performance comparable to that of state of the art non-linear solutions but at increased flexibility and reduced complexity.  Furthermore, the ability to allocate different SQINR constraints to different users depending on their quality of service criterion makes the proposed solution attractive from an operator's of view. The results presented here were for the canonical case with equal target SQINRs for all $K$ users. Further improvement might be possible by choosing unequal target SQINRs. This, however, is a scheduling problem and will be tackled by the operator at a higher level. Lastly, the ideas presented in this paper are generalizable to the setting where the users have multiple antennas by modifying the alternating minimization procedure to include a third step for each user's linear combiner optimization over a pre-defined codebook. The index of the optimized combiner will have to be reported back to the users over a control channel \cite{feedforward}.

%%%%%%%%%%%%%%%%%%%%%%%%%%%%%%%%%%%%%%%%%%%%%%%%%%%%%%%%
\section{Conclusion} \label{sec:conc}
In this paper, we presented a solution to the MU-DL-BF problem under 1-bit hardware constraints at the BS and per-user target SQINR constraints. Our proposed solution jointly optimizes the DL precoders and the power allocated to each user based on the UL-DL duality principle proved under an uncorrelated quantization noise assumption. We generalized the proposed solution to scenarios with correlated quantization noise by adding dummy users to the system which operate in the null space of the true users. Our results demonstrated significant gains over existing linear precoding strategies in terms of the ergodic sum rate and ergodic minimum rate. We also found performance comparable to state of the art non-linear methods in terms of the uncoded BER. The proposed solution provides the flexibility of low-complexity linear precoders for a significantly improved performance and the ability to cater to users with different performance criterion.

 %Our ongoing work in this direction will generalize this work to settings where only statistical knowledge of the user channels is available. 
 
An interesting direction is to extend the proposed solution to constant envelope quantizers and other type of hardware constraints. It is not immediately clear whether the UL-DL duality principle (and the proposed solution based on that) generalize to the quantization noise that results from constant envelope quantizers. Another interesting direction is extending this work to frequency selective wideband channels under an OFDM formulation. The orthogonal nature of the subcarriers, which is one of the main advantages of using OFDM, is destroyed under 1-bit quantization considered in this paper. Hence, the problem has to be revisited from scratch and will be the focus of our work in the next iteration of this paper.

%%%%%%%%%%%%%%%%%%%%%%%%%%%%%%%%%%%%%%%%%%%%%%%%%%%%%%%%
\appendices

%%%%%%%%%%%%%
\section{Proof of lemma \ref{lemma:lemma1}}\label{sec:A}
By dropping the noise term on the RHS of (\ref{eq:downlinkSINR_covar2}), let us define the signal-to-quantization-plus-interference ratio (SQIR) for the  $k^\thh$ user as 

\begin{equation}
\begin{split}
& \hat{\gamma}_k^{\text{DL}}(\mathbf{T} , \mathbf{q}) = \\
& \frac{ q_k \mathbf{t}_k^\Her \mathbf{R}_k \mathbf{t}_k}{  \sum_{ \substack{i = 1\\ i \neq k}}^K q_i \mathbf{t}_i^\Her \mathbf{R}_k \mathbf{t}_i + \left( \frac{\pi}{2} - 1 \right) \tr \left( \sum_{i=1}^K q_i \mathbf{t}_i^\Her \diag(\mathbf{R}_k^\ast) \mathbf{t}_i  \right)}.
\label{eq:A1}
\end{split}
\end{equation}
It can be observed from (\ref{eq:A1}) that the SQIR is a constant function of scalar multiples of the DL power allocation vector $\mathbf{q}$ i.e.

\begin{equation}
\hat{\gamma}_k^{\text{DL}}(\mathbf{T} , \lambda\mathbf{q}) = \hat{\gamma}_k^{\text{DL}}(\mathbf{T} , \mathbf{q}),
\label{eq:A2}
\end{equation}
for all positive $\lambda$. It can also be observed from (\ref{eq:downlinkSINR_covar2}) that the DL SQINR, ${\gamma}_k^{\text{DL}}(\mathbf{T} , \mathbf{q})$, is a monotonically increasing function of scalar multiples of the DL power allocation vector $\mathbf{q}$ i.e.

\begin{equation}
{\gamma}_k^{\text{DL}}(\mathbf{T} , \lambda\mathbf{q}) > {\gamma}_k^{\text{DL}}(\mathbf{T} , \mathbf{q}),
\label{eq:A3}
\end{equation}
for $\lambda > 1$. Furthermore by comparing (\ref{eq:downlinkSINR_covar2}) and (\ref{eq:A1}), it can be seen that

\begin{equation}
\lim_{||\mathbf{q}||_2 \to \infty } {\gamma}_k^{\text{DL}}(\mathbf{T} , \mathbf{q}) = \lim_{||\mathbf{q}||_2 \to \infty } \hat{\gamma}_k^{\text{DL}}(\mathbf{T} , \mathbf{q}).
\label{eq:A4}
\end{equation}
For a target DL SQINR set $[\gamma_1, \dots \gamma_k, \dots \gamma_K]$ to be feasible

\begin{equation} \label{eq:A5}
\begin{split}
1 &\leq \min_{1\leq k \leq K}  \frac{{\gamma}_k^{\text{DL}}(\mathbf{T} , \mathbf{q})}{\gamma_k} \\
   & \overset{(\ref{eq:A3})}{<} \max_{||\mathbf{q}||_2 \to \infty } \left( \min_{1\leq k \leq K}  \frac{{\gamma}_k^{\text{DL}}(\mathbf{T} , \mathbf{q})}{\gamma_k} \right) \triangleq R^\star.
\end{split}
\end{equation}
Making use of (\ref{eq:A2}) and (\ref{eq:A4}), the upper bound (\ref{eq:A5}) is equivalently given by

\begin{equation} \label{eq:A6}
\begin{split}
R^\star & \overset{(\ref{eq:A4})}{=} \max_{||\mathbf{q}||_2 \to \infty } \left( \min_{1\leq k \leq K}  \frac{\hat{\gamma}_k^{\text{DL}}(\mathbf{T} , \mathbf{q})}{\gamma_k} \right) \\
& \overset{(\ref{eq:A2})}{=} \max_{||\mathbf{q}||_2 = 1 } \left(  \min_{1\leq k \leq K}  \frac{\hat{\gamma}_k^{\text{DL}}(\mathbf{T} , \mathbf{q})}{\gamma_k} \right).
\end{split}
\end{equation}
Similar to the optimal DL power allocation in Section \ref{sec:DLpower}, the solution to the optimization problem (\ref{eq:A6}) results in \emph{equal} achieved SQIR to target SQIR ratio for all $K$ users given by

\begin{equation} \label{eq:A7}
R^\star = \frac{\hat{\gamma}_1^{\text{DL}}(\mathbf{T} , \mathbf{q}^\star)}{\gamma_1} = \dots = \frac{\hat{\gamma}_K^{\text{DL}}(\mathbf{T} , \mathbf{q}^\star)}{\gamma_K},
\end{equation}
where $\mathbf{q}^\star$ is the power allocation vector which solves (\ref{eq:A6}). The proof of this claim follows exactly the proof of Lemma \ref{lemma:lemma3} in Appendix \ref{sec:C}. The $K$ equations in (\ref{eq:A7}) can be written in matrix form as 

\begin{equation} \label{eq:A8}
\mathbf{q}^\star \frac{1}{R^\star} = \mathbf{D}(\mathbf{T}) \mathbf{\Psi}(\mathbf{T}) \mathbf{q}^\star.
\end{equation}
It can be observed from (\ref{eq:A8}) that the achieved SQIR to target SQIR balance value, $R^\star$, equals the reciprocal of an eigenvalue of the matrix $ \mathbf{D}(\mathbf{T}) \mathbf{\Psi}(\mathbf{T})$ and the optimal power allocation vector is given by the corresponding eigenvector. It is also known from Perron-Frobenius theory that the optimal eigenvalue/eigenvector pair correspond to the maximal eigenvalue of the non-negative matrix $ \mathbf{D}(\mathbf{T}) \mathbf{\Psi}(\mathbf{T})$. Hence

\begin{equation} \label{eq:A9}
\lambda_\maxx(\mathbf{D}(\mathbf{T}) \mathbf{\Psi}(\mathbf{T})) = \frac{1}{R^\star} \overset{(\ref{eq:A5})}{<} 1.
\end{equation}
This concludes the proof of Lemma \ref{lemma:lemma1}. 

%%%%%%%%%%%%%
\section{Proof of lemma \ref{lemma:lemma2}}\label{sec:B}
Assume $\left(  \mathbf{I}_K -  \mathbf{D}({\mathbf{T}}) \mathbf{\Psi}({\mathbf{T}}) \right)$ is not invertible. This must mean that for some vector $\mathbf{b}$

\begin{equation} \label{eq:B1}
\begin{split}
\left(  \mathbf{I}_K -  \mathbf{D}({\mathbf{T}}) \mathbf{\Psi}({\mathbf{T}}) \right) \mathbf{b} &= \mathbf{0}_K \\
\Rightarrow \mathbf{D}({\mathbf{T}}) \mathbf{\Psi}({\mathbf{T}}) \mathbf{b} &=  \mathbf{b}.
\end{split}
\end{equation}
This implies that the matrix $\mathbf{D}({\mathbf{T}}) \mathbf{\Psi}({\mathbf{T}})$ has an eigenvalue equal to 1. We know from Lemma \ref{lemma:lemma1} that $\lambda_\maxx(\mathbf{D}(\mathbf{T}) \mathbf{\Psi}(\mathbf{T})) < 1$. Hence this is a contradiction and the matrix $\left(  \mathbf{I}_K -  \mathbf{D}({\mathbf{T}}) \mathbf{\Psi}({\mathbf{T}}) \right)$ is invertible for any feasible target DL SQINR set $[\gamma_1, \dots \gamma_k, \dots \gamma_K]$.

%%%%%%%%%%%%%
\section{Proof of lemma \ref{lemma:lemma3}}\label{sec:C}
Let the $i^\thh$ user be such that 

\begin{equation} \label{eq:C1}
\frac{\gamma_i^{\text{DL}}(\mathbf{T^\star , q^\star})}{\gamma_i} > R^\text{DL}_\text{opt}(P_\BS , \mathbf{T}^\star) = \min_{1\leq k \leq K} \frac{\gamma_k^{\text{DL}}(\mathbf{T^\star , q^\star})}{\gamma_k}.
\end{equation}
It can be seen from (\ref{eq:downlinkSINR_covar2}) that the DL SQINR $\gamma_k^{\text{DL}}(\mathbf{T^\star , q})$ is an increasing function of $q_k$ and a decreasing function of $q_\ell$ for $\ell \neq k$. The power allocated to the $i^\thh$ user, $q_i$, can be decreased without reducing the objective function $\min_{1\leq k \leq K} \frac{\gamma_k^{\text{DL}}(\mathbf{T^\star , q^\star})}{\gamma_k}$. This excess power can then be distributed equally amongst the $K$ users. Since $\gamma_k^{\text{DL}}(\mathbf{T} , \alpha \mathbf{q}) > \gamma_k^{\text{DL}}(\mathbf{T} , \mathbf{q})$ for $\alpha > 1$, this would result in a larger optimum value of the objective function $\min_{1\leq k \leq K} \frac{\gamma_k^{\text{DL}}(\mathbf{T^\star , q^\star})}{\gamma_k}$. Consequently, the initial assumption was a contradiction and all $K$ users achieve the same achieved SQINR to target SQINR ratio.

%%%%%%%%%%%%%
\section{Proof of lemma \ref{lemma:lemma5}}\label{sec:D}
It was shown in \cite{MUDL} that for any positive $K$-dimensional vectors $\mathbf{b}$ and $\mathbf{c}$

\begin{equation}  \label{eq:D1}
\max_{\mathbf{x}} \frac{ \mathbf{x}^\T \mathbf{b} }{ \mathbf{x}^\T \mathbf{c} } = \max_{ 1 \leq k \leq K } \frac{ \mathbf{b}_k }{ \mathbf{c}_k }.
\end{equation}
Using (\ref{eq:D1}) and the non-negativity of $\mathbf{\Lambda}(\mathbf{T}, P_\BS) \mathbf{p}_\ext$, it follows that

\begin{equation} \label{eq:D2}
 \hat{\lambda} \left( \mathbf{T}, P_\BS, \mathbf{p}_\ext \right) = \max_{1 \leq k \leq K+1}  \frac{ \mathbf{e}_k^\T \mathbf{\Lambda}(\mathbf{T}, P_\BS) \mathbf{p}_\ext  }{ \mathbf{e}_k^\T \mathbf{p}_\ext }.
\end{equation} 
Using (\ref{eq:uplinkSINR_covar5}) and (\ref{eq:gammaMatrix}), the first $K$ equations in (\ref{eq:D2}) can be written as

\begin{equation} \label{eq:D3}
\max_{1 \leq k \leq K}  \frac{ \mathbf{e}_k^\T \mathbf{\Lambda}(\mathbf{T}, P_\BS) \mathbf{p}_\ext  }{ \mathbf{e}_k^\T \mathbf{p}_\ext } = \max_{1 \leq k \leq K} \frac{\gamma_k}{\gamma_k^{\text{UL}}\left(  \mathbf{t}_k ,  \mathbf{p} \right) }.
\end{equation}
It also follows from (\ref{eq:uplinkSINR_covar5}) and (\ref{eq:gammaMatrix}) that

\begin{equation} \label{eq:D4}
\begin{split}
\frac{ \mathbf{e}_{K+1}^\T \mathbf{\Lambda}(\mathbf{T}, P_\BS) \mathbf{p}_\ext  }{ \mathbf{e}_{K+1}^\T \mathbf{p}_\ext } & = \frac{1}{P_\BS} \sum_{k=1}^K  \frac{\mathbf{p}_k \gamma_k}{\gamma_k^{\text{UL}}\left(  \mathbf{t}_k ,  \mathbf{p} \right) }\\
& \overset{(a)}{\leq}  \left( \max_{1 \leq k \leq K} \frac{\gamma_k}{\gamma_k^{\text{UL}}\left(  \mathbf{t}_k ,  \mathbf{p} \right) } \right) \frac{1}{P_\BS} \sum_{k=1}^K  \mathbf{p}_k \\
& =  \max_{1 \leq k \leq K} \frac{\gamma_k}{\gamma_k^{\text{UL}}\left(  \mathbf{t}_k ,  \mathbf{p} \right) }.
\end{split}
\end{equation}
(a) follows because the max is greater than the average. This shows that the $(K+1)^\thh$ equation in (\ref{eq:D2}) is smaller than or equal to the first $K$ equations. Hence

\begin{equation} \label{eq:D5}
 \hat{\lambda} \left( \mathbf{T}, P_\BS, \mathbf{p}_\ext \right) = \max_{1 \leq k \leq K} \frac{\gamma_k}{\gamma_k^{\text{UL}}\left(  \mathbf{t}_k ,  \mathbf{p} \right) }.
\end{equation} 

%%%%%%%%%%%%%%%%%%%%%%%%%%%%%%%%%%%%%%%%%%%%%%%%%%%%%%%%
\bibliographystyle{IEEEbib}
\bibliography{ref}

\end{document}